\documentclass[12pt]{article}

\usepackage[margin=0.8in]{geometry}

\usepackage{float, enumitem, booktabs}
\usepackage{caption, subcaption}
\usepackage{graphicx, enumerate, bm}

\usepackage{amsfonts, amsmath, amsthm, mathtools}
\usepackage[version=4]{mhchem}
\usepackage{siunitx, url, colortbl, array, fancyhdr}
\usepackage{longtable,tabularx}
\usepackage{tikz}
\usetikzlibrary{patterns,angles,quotes,arrows,math}
\newtheoremstyle{mystyle}
  {}
  {}
  {}
  {}
  {\bfseries}
  {:}
  { }
  {}
\newtheorem*{remark}{Remark}
\newtheorem{definition}{Definition}
\newtheorem{problem}{Problem}
\newtheorem{theorem}{Theorem}
\newtheorem{proposition}{Proposition}
\newtheorem{lemma}{Lemma}

\DeclareMathOperator{\Image}{Im}
\DeclareMathOperator{\Ker}{Ker}
\DeclareMathOperator{\rank}{rank}
\DeclareMathOperator{\relint}{relint}
\DeclareMathOperator{\Span}{span}
\DeclareMathOperator{\interior}{int}
\DeclareMathOperator{\co}{co}
\DeclareMathOperator{\atan2}{atan2}
\renewcommand*{\arraystretch}{0.7}
\usepackage[numbers]{natbib}

\title{Delayed resilient trajectory tracking after partial loss of control authority over actuators}

\author{Jean-Baptiste Bouvier\thanks{Department of Aerospace Engineering and Coordinated Science Laboratory, University of Illinois at Urbana-Champaign, Urbana, IL 61801, USA.
  (bouvier3@illinois.edu).}
\and Himmat Panag\footnotemark[1] \and Robyn Woollands\footnotemark[1] \and Melkior Ornik\footnotemark[1]
}

\begin{document}

\maketitle

\begin{abstract}
    After the loss of control authority over thrusters of the Nauka module, the International Space Station lost attitude control for 45 minutes with potentially disastrous consequences. Motivated by a scenario of orbital inspection, we consider a similar malfunction occurring to the inspector satellite and investigate whether its mission can still be safely fulfilled.
    While a natural approach is to counteract in real-time the uncontrolled and undesirable thrust with the remaining controlled thrusters, vehicles are often subject to actuation delays hindering this approach.
    Instead, we extend resilience theory to systems suffering from actuation delay and build a resilient trajectory tracking controller with stability guarantees relying on a state predictor. We demonstrate that this controller can track accurately the reference trajectory of the inspection mission despite the actuation delay and the loss of control authority over one of the thrusters.
\end{abstract}

\section*{Notation}

For a set $\mathcal{X} \subseteq \mathbb{R}^n$, $\dim(\mathcal{X})$ denotes its dimension, $\partial \mathcal{X}$ its boundary, $\co(\mathcal{X})$ its convex hull, $\interior(\mathcal{X})$ its interior, and $\relint(\mathcal{X})$ its relative interior as defined in \citep{inf_dim_analysis}. The Minkowski addition of $\mathcal{X}$ and $\mathcal{Y} \subseteq \mathbb{R}^n$ is denoted by $\mathcal{X} \oplus \mathcal{Y} := \big\{x + y : (x, y) \in \mathcal{X} \times \mathcal{Y}\big\}$ and their Minkowski difference is $\mathcal{X} \ominus \mathcal{Y} := \big\{ z \in \mathbb{R}^n : z + y \in \mathcal{X}\ \text{for all}\ y \in \mathcal{Y}\big\}$. The set of all functions $f : [0, +\infty) \rightarrow \mathcal{X}$ is denoted by $\mathcal{F}(\mathcal{X})$.
The closed ball of dimension $b$, radius $r \geq 0$, and center $c$ is denoted $\mathbb{B}^b(c,r) := \big\{ x \in \mathbb{R}^b : \|x - c\| \leq r\big\}$.
For a set $\Lambda \subseteq \mathbb{C}$, we say that $Re(\Lambda) \leq 0$ (resp. $Re(\Lambda) = 0$) if the real part of each $\lambda \in \Lambda$ verifies $Re(\lambda) \leq 0$ (resp. $Re(\lambda) = 0$).
The norm of a matrix $A$ is $\|A\| := \underset{x\, \neq\, 0}{\sup} \frac{\|Ax\|}{\|x\|} = \underset{\|x\|\, =\, 1}{\max} \|Ax\|$ and the set of its eigenvalues is $\Lambda(A)$. If $A$ is positive definite, $A \succ 0$, then we define $\lambda_{min}^A := \min \big\{ \Lambda(A) \big\}$, $\lambda_{max}^A := \max \big\{ \Lambda(A) \big\}$ and the associated norm $\|x\|_A := \sqrt{x^\top A x}$.
The controllability matrix of pair $(A,B)$ is $\mathcal{C}(A,B) = \big[B\, A \hspace{-0.4mm}B\, \hdots\, A^{n-1}\hspace{-0.8mm}B\big]$. The linear span of set $\mathcal{X}$ is denoted by $\Span(\mathcal{X}) := \big\{ \sum_{i=1}^k \alpha_i x_i,\ k \in \mathbb{N},\ \alpha_i \in \mathbb{R},\ x_i \in \mathcal{X} \big\}$.

\section{Introduction}

With an increase in the number of active satellites, there is a growing demand for on-orbit satellite inspection, e.g., to assess damage on satellites, prevent unnecessary spacewalks of astronauts, or enforce the ban of space weapons \citep{inspection_thesis, attitude_inspection, inspection}. The importance of satellite inspection is also reflected by the creation of spacecraft entirely dedicated to on-orbit inspections, like the robot Laura from the Rogue Space Systems Corporation\footnote{\url{https://rogue.space/orbots/}}. 

Partly inspired by the on-orbit servicing Restore-L mission \citep{Restore-L}, our scenario of interest consists of an \textit{orbital inspection} of a satellite by a spacecraft that completes a full revolution around the target satellite. 
Following an on-board computer error, the inspecting spacecraft endures a \textit{loss of control authority} \citep{IFAC} over one of its thrusters, similarly to what happened to the Nauka module when docked to the International Space Station \citep{ISS_thruster}. This malfunction consists in one of the thrusters producing uncontrolled and thus possibly undesirable thrust with the same capabilities as before the malfunction. 

Classically, changing or unknown dynamics are studied through robust, adaptive, and fault-tolerant control theories. However, robust control needs the undesirable thrust to be significantly smaller than the controlled thrust \citep{weak_robust_control}, which is not the case in our scenario.
In turn, the estimators of adaptive controllers are unlikely to converge in the presence of uncontrolled thrust varying quickly \citep{weak_robust_control}.
Fault-tolerant theory provides a wider framework, but actuator failure investigations are usually limited either to completely disabled thrusters \citep{abort-safe, Breger}, to actuators ``locking in place" and producing constant inputs \citep{actuator_lock}, or to actuators with reduced effectiveness \citep{loss_control_effectiveness, Fault_Tolerant_Review}. Since the uncontrolled thruster under study can still produce a full range of thrust, this malfunction is not covered by existing fault-tolerant theory \citep{Fault_Tolerant_Review}.
Instead, we adopt the \textit{resilience} framework \citep{IFAC, TAC, SIAM_CT} by assuming the presence of thrust sensors, thus enabling the implementation of fault-detection and isolation methods \citep{thruster_delay}. 

Because of the reaction times of the sensors and thrusters \citep{thruster_delay}, the controller is likely not able counteract the undesirable thrust in real-time. Our objective is then to develop a control strategy for safely carrying out the inspection mission despite the malfunctioning thruster and the actuation delay. More specifically, we want the damaged spacecraft to accurately follow a safe reference trajectory. In our simulation, we choose a minimal-fuel reference trajectory generated by the convex optimization method of \citep{Ortolano}. 

In the preliminary version of this work \citep{AAS}, resilience was not established analytically but empirically with a Monte-Carlo simulation. We address this issue by extending resilience theory to linear systems with actuation delay and by establishing a resilient trajectory tracking controller. Entirely new material includes Sections~\ref{subsec: tracking no delay}, \ref{sec:theory delay}, \ref{sec:application} and \ref{sec:simulation}.

The main contributions of this work are fourfold. Firstly, we establish the resilience of a spacecraft with nonlinear dynamics. Secondly, we extend resilience theory to systems with actuation delay and provide sufficient conditions for these systems to be able to perform trajectory tracking.
Thirdly, we build a resilient trajectory tracking controller with guaranteed performance for the nonlinear spacecraft dynamics.
Finally, we demonstrate that on-orbit inspection can be performed safely despite actuation delay and a loss of control authority over a thruster.

The remainder of this paper is structured as follows. Section~\ref{sec:motivation} introduces our problem of interest and the relative dynamics of the satellites. In Section~\ref{sec:theory NO delay}, we ignore actuation delay to apply existing resilience theory to the malfunctioning spacecraft to demonstrate its remaining capabilities in terms of resilient reachability and trajectory tracking.
Then, in Section~\ref{sec:theory delay} we establish novel results to extend resilience theory to linear systems suffering from actuation delay.
Section~\ref{sec:application} builds on this theory to produce a resilient trajectory tracking controller with guaranteed performance despite actuation delay.
Finally, Section~\ref{sec:simulation} implements this controller in a numerical simulation of the inspection mission.

\section{Motivation and background}\label{sec:motivation}

We consider two spacecraft on circular orbit around Earth. The mission of the chaser spacecraft is to inspect the target spacecraft. As we are interested in proximity operations, we employ the Clohessy-Wiltshire equations in a local-vertical, local-horizontal frame \citep{Ortolano}. The state vector $X = \big( x\ y\ z\ \dot x\ \dot y\ \dot z \big) \in \mathbb{R}^6$ represents the difference in position and velocity between the two spacecraft and initially follows the dynamics
\begin{equation}\label{eq:ODE}
    \dot X(t) = AX(t) + r\bar{B} \bar{u}(t), \qquad X(0) = X_0 \in \mathbb{R}^6,
\end{equation}
with a \textit{thrust-to-mass ratio} $r = 1.5 \times 10^{-4}\, m/s^2$, as we consider a chaser spacecraft of mass $600\, kg$ and five PPS-1350 thrusters of maximal thrust $90\, mN$ \citep{PPS-1350} controlled by the inputs $\bar{u} = \big( \bar{u}_1\ \bar{u}_2\ \bar{u}_3\ \bar{u}_4\ \bar{u}_5 \big) \in [0, 1]^5$.
Because the $z$-dynamics of the Clohessy-Witshire equations are decoupled from the other two axis, we focus on the two-dimensional dynamics in the $(x, y)$-plane, with matrices:
\begin{equation*}\label{eq:2D}
    A = \begin{bmatrix} 0 & 0 & 1 & 0 \\ 0 & 0 & 0 & 1 \\ 3 \Omega^2 & 0 & 0 & 2 \Omega \\ 0 & 0 & -2 \Omega & 0 \end{bmatrix} \qquad \text{and} \qquad \bar{B} = \begin{bmatrix} 0 & 0 & 0 & 0 & 0 \\
     0 & 0 & 0 & 0 & 0 \\ 1 & 1 & -1 & -\sqrt{2} & -1 \\ 1 & -1 & -1 & 0 & 1 \end{bmatrix},
\end{equation*}
where $\Omega = 0.00106\, s^{-1}$ is the mean orbital rate of the target's orbit.
The thrusters do not create any torque \citep{abort-safe} since they are rigidly fixed on the spacecraft and are aligned with its center of mass, as illustrated on Fig.~\ref{fig:spacecraft}. To perform its inspection mission, the chaser spacecraft relies on a fixed camera constantly pointing at the target thanks to the reaction wheels controlling the attitude of the chaser, as shown on Fig.~\ref{fig:spacecraft}. Because of these attitude changes, the relative dynamics lose their linearity to become
\begin{equation}\label{eq:ODE rotation}
    \dot X(t) = A X(t) + r R_\theta(t) \bar{B} \bar{u}(t), \qquad \text{with} \qquad R_\theta(t) = \begin{bmatrix} 1 & 0 & 0 & 0 \\ 0 & 1 & 0 & 0 \\ 0 & 0 & \cos\big( \theta(t) \big) & -\sin\big( \theta(t) \big) \\ 0 & 0 & \sin\big( \theta(t) \big) & \cos\big( \theta(t) \big) \end{bmatrix},
\end{equation}
where $\theta$ is illustrated on Fig.~\ref{fig:spacecraft} and is defined as the 2-argument arctangent $\theta(t) := \atan2 \big( y(t),\, x(t) \big)$.
\begin{figure}[htbp!]
    \centering
    \begin{tikzpicture}[scale = 0.9]
        \draw[->] (0, 0) -- (7, 0);
        \node at (7, 0.3) {$x$};
        \draw[->] (0, 0) -- (0, 4);
        \node at (0.3, 4) {$y$};
        
        \draw[dashed, rotate = 23] (0, 0) -- (5.65, 0);
        \draw[->, very thick] (4, 0) arc (0:23:4);
        \node at (4.1, 0.8) {$\theta$};
        
        \begin{scope}[rotate=23, shift={(7,0)}, scale = 0.3]
            \draw[very thick] (3, 1.8) -- (3, -1.8) -- (2.8, -2) -- (-2.8, -2) -- (-3, -1.8) -- (-3, 1.8) -- (-2.8, 2) -- (2.8, 2) -- (3, 1.8);
            
            \draw[very thick, fill] (-3.6, 0.6) .. controls (-3, 0.2) and (-3, -0.2) .. (-3.6, -0.6) -- (-3.6, 0.6);
            \draw[very thick] (-3, 0.1) -- (-4.2, 0) -- (-3, -0.1);
            \draw[fill] (-4.2, 0) circle (2.5pt);
            
            \draw[very thick] (3, 0.2) -- (3.5, 0.5) -- (3.5, -0.5) -- (3, -0.2);
            \draw[very thick] (2.8, 2) -- (3, 2.6) -- (3.6, 1.9) -- (3, 1.8);
            \draw[very thick] (2.8, -2) -- (3, -2.6) -- (3.6, -1.9) -- (3, -1.8);
            \draw[very thick] (-2.8, 2) -- (-3, 2.6) -- (-3.6, 1.9) -- (-3, 1.8);
            \draw[very thick] (-2.8, -2) -- (-3, -2.6) -- (-3.6, -1.9) -- (-3, -1.8);
            \draw[->, red, ultra thick] (3.5, 0) -- (4.7, 0);
            \node at (4.8, 0.8) {$\bar{u}_4$};
            \draw[->, red, ultra thick] (3.3, 2.25) -- (4.2, 2.9);
            \node at (4.9, 2.9) {$\bar{u}_3$};
            \draw[->, red, ultra thick] (-3.3, -2.25) -- (-4.2, -2.9);
            \node at (-4.9, -2.8) {$\bar{u}_1$};
            \draw[->, red, ultra thick] (-3.3, 2.25) -- (-4.2, 2.9);
            \node at (-4.9, 2.9) {$\bar{u}_2$};
            \draw[->, red, ultra thick] (3.3, -2.25) -- (4.2, -2.9);
            \node at (4.9, -2.8) {$\bar{u}_5$};
            \node at (0, 0) {chaser};
        \end{scope}
        
        \filldraw[white] (0.7, 0.4) -- (-0.7, 0.4) -- (-0.7, -0.4) -- (0.7, -0.4);
        \draw[very thick] (0.7, 0.4) -- (-0.7, 0.4) -- (-0.7, -0.4) -- (0.7, -0.4) -- (0.7, 0.4);
        \node at (0, 0) {target};
    \end{tikzpicture}
    \caption{Relative positions and attitudes of the two satellites, with the camera of the chaser always pointing at the target thanks to an independent attitude control system.}
    \label{fig:spacecraft}
\end{figure}
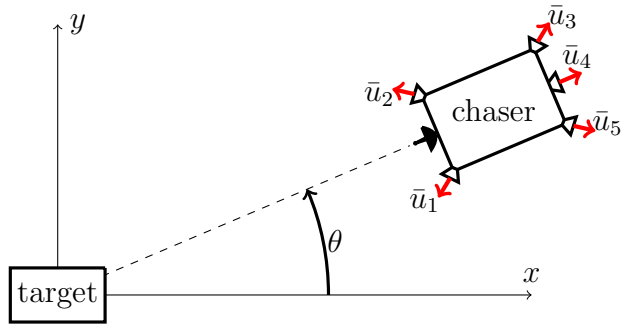

Following the Restore-L protocol \citep{Restore-L}, we assume that the chaser must come within $80\, m$ of the target for a precise optical inspection. Hence, we want the chaser to occupy successively the 5 following holding points $(0, 80)$, $(-80, 0)$, $(0, -80)$, $(80, 0)$ and $(0, 80)$. Using the convex optimization method \citep{Ortolano} we compute on Fig.~\ref{fig:ref orbit} the minimal fuel trajectory linking these waypoints with 90 minutes transfers. For safety considerations, we consider a keep-out sphere (KOS) of radius $R_{KOS} = 50\, m$ around the target as in the Restore-L mission \citep{Restore-L}.

\begin{figure}[htbp!]
    \centering
    \includegraphics[scale = 0.6]{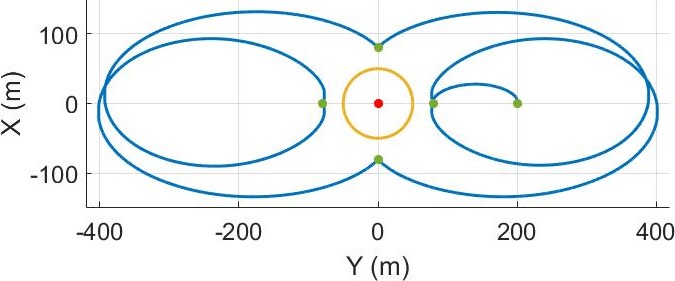}
    \caption{Reference minimal-fuel trajectory (blue) linking the four waypoints (green) to inspect the target satellite (red) without breaching the KOS (yellow).}
    \label{fig:ref orbit}
\end{figure}

Because the chaser is constantly pointing its camera towards the target, its orientation angle $\theta$ (see Fig.~\ref{fig:spacecraft}) varies throughout the trajectory as shown on Fig.~\ref{fig:reference}(\subref{fig:attitude}) and starts at $\theta(0) = 90^\circ$ since the initial position of the spacecraft is on the $y$-axis, as illustrated in Fig.~\ref{fig:ref orbit}.
\begin{figure}[htbp!]
    \centering
    \begin{subfigure}[]{0.46\textwidth}
        \includegraphics[scale=0.43]{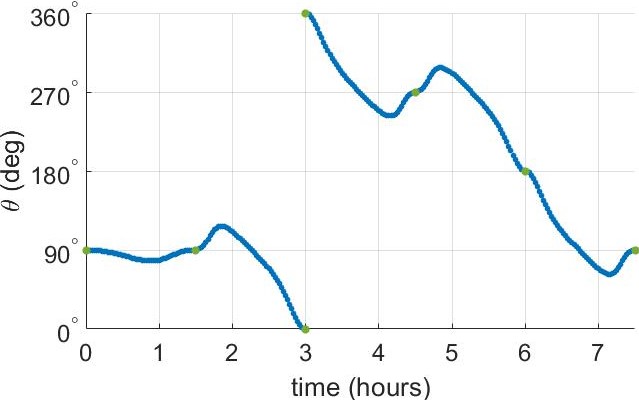}
        \caption{Orientation of the chaser $\theta$ during its mission, with the waypoints in green.}
        \label{fig:attitude}
    \end{subfigure}\hfill
    \begin{subfigure}[]{0.53\textwidth}
        \includegraphics[scale = 0.48]{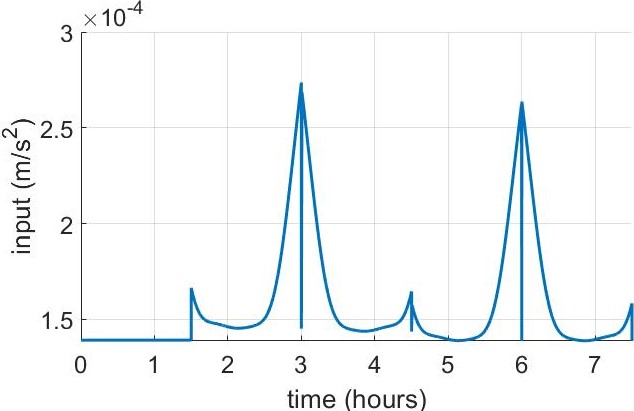}
        \caption{Combined reference thrust signal $\|u_\text{ref}\|$.}
        \label{fig:ref_input}
    \end{subfigure}
    \caption{Chaser orientation and thrust profile for the reference trajectory.}
    \label{fig:reference}
\end{figure}

When the chaser follows the minimal-fuel reference trajectory of Fig.~\ref{fig:ref orbit}, the resulting thrust profile is represented on Fig.~\ref{fig:reference}(\subref{fig:ref_input}) and shows several impulses. Their symmetry comes from the symmetry of the reference trajectory of Fig.~\ref{fig:ref orbit}. This thrust profile is obtained from the convex optimization method \citep{Ortolano} by propagating dynamics \eqref{eq:ODE rotation} along the fuel-optimal trajectory.

Having described the nominal dynamics and the reference trajectory, we now study the malfunction impacting the chaser. Similarly to what happened to the Nauka module docked to the ISS \citep{ISS_thruster}, we assume that an error in the on-board computer of the chaser satellite causes the controller to lose authority over one of the thrusters. The input signal $\bar{u} \in \mathcal{F}(\bar{\mathcal{U}})$ of \eqref{eq:ODE rotation} is then split between the undesirable signal $w \in \mathcal{F}(\mathcal{W})$, $\mathcal{W} = [0, 1]$ and the controlled signal $u \in \mathcal{F}(\mathcal{U})$, $\mathcal{U} = [0, 1]^4$. Matrix $\bar{B}$ is accordingly split into two constant matrices $B \in \mathbb{R}^{4 \times 4}$ and $C \in \mathbb{R}^4$ so that the dynamics of the malfunctioning satellite become
\begin{equation}\label{eq:split ODE rotation}
    \dot X(t) = AX(t) + r R_\theta(t) Bu(t) + r R_\theta(t) Cw(t), \qquad X(0) = X_0 \in \mathbb{R}^4.
\end{equation}
We can then formulate our problem of interest.

\begin{problem}\label{prob:mission}
    With what accuracy can the chaser satellite track the reference trajectory even after enduring a loss of control authority over any one of its thrusters?
\end{problem}

To address Problem~\ref{prob:mission}, we start by determining over which thrusters the spacecraft can resiliently lose control. This investigation is carried out in Section~\ref{sec:theory NO delay} within the resilience framework of \citep{IFAC, ECC} based on the 'snap decision rule' of \citep{Hajek} where the controller $u(t)$ has instantaneous knowledge of the state $X(t)$ and of the uncontrolled input $w(t)$.
This assumption is lifted in subsequent sections where we study resilient trajectory tracking despite actuation delay.

\section{Spacecraft resilience with instantaneous control}\label{sec:theory NO delay}

We rely on \textit{resilience theory} \citep{ECC} to determine the thrusters over which the chaser spacecraft can lose control while remaining capable of accomplishing its mission.

\subsection{Resilient reachability}\label{subsec: spacecraft resilience}

Let us first recall the notion of \textit{resilience} from \citep{ECC} adapted to system~\eqref{eq:ODE rotation}.
\begin{definition}\label{def: resilience}
    System~\eqref{eq:ODE rotation} is \emph{resilient} to the loss of control authority over one of its thrusters if for any target $X_{goal} \in \mathbb{R}^4$ and any undesirable signal $w \in \mathcal{F}(\mathcal{W})$ there exists a control signal $u \in \mathcal{F}(\mathcal{U})$ such that the resulting malfunctioning system~\eqref{eq:split ODE rotation} can reach $X_{goal}$ in finite time.
\end{definition}
Resilience is not automatically sufficient to complete the mission of Problem~\ref{prob:mission} since Definition~\ref{def: resilience} only concerns target reachability and not trajectory tracking.
However, resilience is necessary for mission completion because tracking is impossible without some degree of reachability. Following the method of \citep{Hajek}, to assess the resilience of system~\eqref{eq:ODE rotation}, we introduce the associated dynamics
\begin{equation}\label{eq:Hajek ODE}
    \dot X(t) = AX(t) + r R_\theta(t) p(t), \qquad X(0) = X_0, \qquad p(t) \in \mathcal{P} := B\mathcal{U} \ominus (-C\mathcal{W}),
\end{equation}
where $\mathcal{P}$ is the Minkowski difference between the set of controlled inputs $B\mathcal{U} := \big\{ Bu : u \in \mathcal{U} \big\}$ and the negative of the set of undesirable inputs $-C\mathcal{W} := \big\{-Cw : w \in \mathcal{W} \big\}$, i.e., $\mathcal{P} = \big\{ p \in B\mathcal{U} : p - Cw \in B\mathcal{U}\ \text{for all}\ w \in \mathcal{W} \big\}$. Then, $\mathcal{P}$ represents the amount of control authority remaining after counteracting the worst undesirable input. 

According to \citep{ECC_extended}, if systems~\eqref{eq:ODE rotation} and \eqref{eq:Hajek ODE} were linear, the resilience of system~\eqref{eq:ODE rotation} would be equivalent to the controllability of system~\eqref{eq:Hajek ODE}. Let us now define \textit{controllability} in parallel with \textit{stabilizability} as it will be needed later on.
\begin{definition}
    System~\eqref{eq:Hajek ODE} is \emph{controllable} (resp. \emph{stabilizable}) if for all $X_0 \in \mathbb{R}^4$ and all $X_{goal} \in \mathbb{R}^4$, there exists a time $T$ and a control signal $p \in \mathcal{F}(\mathcal{P})$ driving the state of system~\eqref{eq:Hajek ODE} from $X(0) = X_0$ to $X(T) = X_{goal}$ (resp. to $X(T) = 0$).
\end{definition}

However, nonlinear factor $R_\theta(t)$ in systems~\eqref{eq:ODE rotation} and \eqref{eq:Hajek ODE} prevents us from immediately applying the resilience results of \citep{Hajek, ECC_extended} to these systems. Instead, we provide a partial extension of H\'ajek's duality theorem \citep{Hajek} to nonlinear dynamics.

\begin{theorem}\label{thm: Hajek nonlinear}
    If system~\eqref{eq:Hajek ODE} is controllable, then system~\eqref{eq:ODE rotation} is resilient.
\end{theorem}
\begin{proof}
    Let $X_0 \in \mathbb{R}^4$, $X_{goal} \in \mathbb{R}^4$ and $w \in \mathcal{F}(\mathcal{W})$. Since system~\eqref{eq:Hajek ODE} is controllable, there exists $T \geq 0$ and $p \in \mathcal{F}(\mathcal{P})$ driving the state of system~\eqref{eq:Hajek ODE} from $X_0$ to $X_{goal}$ in time $T$. By definition of $\mathcal{P}$ there exists $u \in \mathcal{F}(\mathcal{U})$ such that $Bu(t) = p(t) -Cw(t)$ for all $t \in [0, T]$. Then, applying input signals $u$ and $w$ to system~\eqref{eq:split ODE rotation} drives its state from $X_0$ to $X_{goal}$ in time $T$. Thus, for all $X_{goal} \in \mathbb{R}^4$ and all $w \in \mathcal{F}(\mathcal{W})$ there exists $u \in \mathcal{F}(\mathcal{U})$ driving malfunctioning system~\eqref{eq:split ODE rotation} to $X_{goal}$ in finite time, i.e., system~\eqref{eq:ODE rotation} is resilient.
\end{proof}

\begin{remark}
    The proof of Theorem~\ref{thm: Hajek nonlinear} is actually valid for all nonlinear systems of the form $\dot x(t) = f\big(t, x(t)\big) + g\big(t, x(t)\big) \bar{B} \bar{u}(t)$, and not just for system~\eqref{eq:ODE rotation}. The reverse implication of Theorem~\ref{thm: Hajek nonlinear} remains an open question for nonlinear dynamics.
\end{remark}

Relying on Theorem~\ref{thm: Hajek nonlinear}, we will now investigate whether system~\eqref{eq:ODE rotation} is resilient to a loss of control authority over thruster no. 4. Indeed, this thruster plays a special role in the actuation of the chaser spacecraft due to its location shown on Fig.~\ref{fig:spacecraft} and yields
\begin{equation}\label{eq:B and C}
    B = \left[\def\arraystretch{0.6}\begin{array}{cccc} 0 & 0 & 0 & 0 \\ 0 & 0 & 0 & 0 \\
     1 & 1 & -1 & -1 \\ 1 & -1 & -1 & 1 \end{array}\right] \qquad \text{and} \qquad C = \left[\def\arraystretch{0.6}\begin{array}{c} 0 \\ 0 \\ -\sqrt{2} \\ 0 \end{array}\right].
\end{equation}
The polytopes $B\mathcal{U}$ and $-C\mathcal{W}$ are both in $\mathbb{R}^4$, but since they are of dimension $2$, we only represent these last two dimensions in Fig.~\ref{fig:sets_BU_CW}. Similarly, the Minkowski difference $\mathcal{P}$ is of dimension $2$ and is also illustrated in Fig.~\ref{fig:sets_BU_CW}.

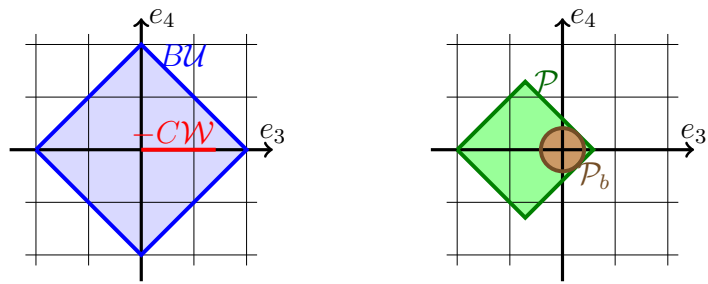
\begin{figure}[htbp!]
    \centering
    \begin{tikzpicture}[scale = 0.7]
        \fill[blue!15!white] (-2, 0) -- (0, -2) -- (2, 0) -- (0, 2) -- (-2, 0);
        
        \draw[ultra thin] (-2.2, 2) -- (2.2, 2);
        \draw[ultra thin] (-2.2, 1) -- (2.2, 1);
        \draw[very thick, ->] (-2.5, 0) -- (2.5, 0);
        \node at (2.5, 0.3) {$e_3$};
        \draw[ultra thin] (-2.2, -1) -- (2.2, -1);
        \draw[ultra thin] (-2.2, -2) -- (2.2, -2);
        
        \draw[ultra thin] (-2, 2.2) -- (-2, -2.2);
        \draw[ultra thin] (-1, 2.2) -- (-1, -2.2);
        \draw[very thick, <-] (0, 2.5) -- (0, -2.5);
        \node at (0.4, 2.5) {$e_4$};
        \draw[ultra thin] (1, 2.2) -- (1, -2.2);
        \draw[ultra thin] (2, 2.2) -- (2, -2.2);
    
        \draw[red, line width=1.5pt] (0, 0) --  (1.4142, 0);
        \node at (0.6, 0.3) {\textcolor{red}{$-C\mathcal{W}$}};
        \draw[draw=blue, line width=1.5pt] (-2, 0) -- (0, -2) -- (2, 0) -- (0, 2) -- (-2, 0);
        \node at (0.8, 1.8) {\textcolor{blue}{$B\mathcal{U}$}};

         \fill[green!40!white] (6, 0) -- (8-0.707, -1.29289) -- (8.5857, 0) -- (8-0.707, 1.29289)  -- (6, 0);
         \fill[brown!80!white] (8, 0) circle (0.41);
        
        \draw[ultra thin] (5.8, 2) -- (10.2, 2);
        \draw[ultra thin] (5.8, 1) -- (10.2, 1);
        \draw[very thick, ->] (5.5, 0) -- (10.5, 0);
        \node at (10.5, 0.3) {$e_3$};
        \draw[ultra thin] (5.8, -1) -- (10.2, -1);
        \draw[ultra thin] (5.8, -2) -- (10.2, -2);
        
        \draw[ultra thin] (6, 2.2) -- (6, -2.2);
        \draw[ultra thin] (7, 2.2) -- (7, -2.2);
        \draw[very thick, <-] (8, 2.5) -- (8, -2.5);
        \node at (8.4, 2.5) {$e_4$};
        \draw[ultra thin] (9, 2.2) -- (9, -2.2);
        \draw[ultra thin] (10, 2.2) -- (10, -2.2);
        
        \draw[green!50!black, line width = 1.5pt] (6, 0) -- (8-0.707, -1.29289) -- (8.5857, 0) -- (8-0.707, 1.29289)  -- (6, 0);
        \node at (7.7, 1.3) {\textcolor{green!40!black}{$\mathcal{P}$}};
        
        \draw[brown!60!black, line width = 1.5pt] (8, 0) circle (0.41);
        \node at (8.6, -0.5) {\textcolor{brown!60!black}{$\mathcal{P}_b$}};
    \end{tikzpicture}
    \caption{2D projection of $B\mathcal{U}$ (blue), $-C\mathcal{W}$ (red), their Minkowski difference $\mathcal{P}$ (green) and the largest ball $\mathcal{P}_b$ (brown) centered at $0$ that fits inside $\mathcal{P}$ for the case of the loss of control authority over thruster no. 4.}
    \label{fig:sets_BU_CW}
\end{figure}

To prove the resilience of system~\eqref{eq:ODE rotation} to a loss of control authority over thruster no. 4, we need to verify the controllability of nonlinear system~\eqref{eq:Hajek ODE}. However, this verification of controllability is generally a difficult problem \citep{Sussmann}. Instead, we will construct a related linear time-invariant system whose controllability implies that of system~\eqref{eq:Hajek ODE}.

\begin{proposition}\label{prop:resilience}
    System~\eqref{eq:ODE rotation} is resilient to a loss of control authority over thruster no. 4.
\end{proposition}
\begin{proof}
    Following Theorem~\ref{thm: Hajek nonlinear}, we will prove controllability of system~\eqref{eq:Hajek ODE} to obtain resilience of system~\eqref{eq:ODE rotation}.
    Because $-C\mathcal{W} \subseteq \interior\big(B\mathcal{U}\big)$, we have $0 \in \interior\big(\mathcal{P}\big)$, as seen on Fig.~\ref{fig:sets_BU_CW}. Then, we can define $\rho_{max}$ as the radius of the largest ball of dimension 2 centered at $0$ and fitting inside $\mathcal{P}$, i.e., $\rho_{max} := \max\big\{ \rho \geq 0 : \mathbb{B}^2(0, \rho) \subseteq \mathcal{P} \big\}$. In our case $\rho_{max} = \sqrt{2} - 1 = 0.414$. Then, the ball $\mathcal{P}_b := \mathbb{B}^2(0, \rho_{max})$ is a subset of $\mathcal{P}$ as illustrated on Fig.~\ref{fig:sets_BU_CW}. Because $\mathcal{P}_b$ is a ball, there is a one-to-one correspondence between inputs $p \in \mathcal{P}_b$ and $R_\theta p \in \mathcal{P}_b$, so the dynamics~\eqref{eq:Hajek ODE} with inputs constrained to $\mathcal{P}_b$ are in fact
    \begin{equation}\label{eq:ODE PBall}
        \dot X(t) = AX(t) + r\hat{B} p(t), \qquad p \in \mathcal{P}_b \subset \mathbb{R}^2, \qquad \hat{B} = \left[\begin{smallmatrix} 0_{2 \times 2} \\ I_2 \end{smallmatrix}\right].
    \end{equation}
    Because the first two rows of $B$ and $C$ defined in Eq.~\eqref{eq:B and C} are null, the geometrical work we completed above only concerns the last two coordinates of the inputs, which explains the structure of matrix $\hat{B}$ in \eqref{eq:ODE PBall}.
    To prove the controllability of system~\eqref{eq:ODE PBall}, we verify the conditions of Corollary 3.7 of \citep{Brammer}:
    \begin{itemize}[topsep=0pt, itemsep=0pt, partopsep=0pt]
        \item $0 \in \mathcal{P}_b$, so taking $p = 0$ makes $\hat{B}p = 0$;
        \item the convex hull of $\mathcal{P}_b$ has a non-empty interior in $\mathbb{R}^2$;
        \item  $\big[ \hat{B}, A \hat{B} \big] = \left[ \begin{smallmatrix} 0_{2 \times 2} & I_2 \\ I_2 & * \end{smallmatrix} \right]$, so $\rank\big( \big[ \hat{B}, A \hat{B} \big] \big) = 4$, i.e., the controllability matrix has full rank;
        \item the real eigenvectors of $A^\top$ are all scalar multiples of $v = (2\Omega, 0, 0, 1)$, which makes $v^\top \hat{B} p = p_2$ for all $p = (p_1, p_2) \in \mathcal{P}_b$ and $p_2$ can be chosen positive or negative since $\rho_{max} > 0$;
        \item the eigenvalues of $A$ are $\big\{ 0, 0, \pm j\Omega \big\}$, so they all have a zero real part.
    \end{itemize}
    Hence, system~\eqref{eq:ODE PBall} is controllable. Because system~\eqref{eq:Hajek ODE} follows the same dynamics as \eqref{eq:ODE PBall}, and has a larger input set encompassing that of system~\eqref{eq:ODE PBall}, it is also controllable. Then, according to Theorem~\ref{thm: Hajek nonlinear} system~\eqref{eq:ODE rotation} is resilient to the loss of control over thruster no. 4.
\end{proof}

We can now proceed to the case of the other four thrusters. Because of their symmetric placement as shown on Fig.~\ref{fig:spacecraft}, we only need to study one thruster and similar conclusions will hold for the others. For the loss of control authority over thruster no. 1, we represent dimensions $3$ and $4$ of polytopes $B\mathcal{U}$, $-C\mathcal{W}$, and $\mathcal{P}$ on Fig.~\ref{fig:sets_BU_CW_1}.

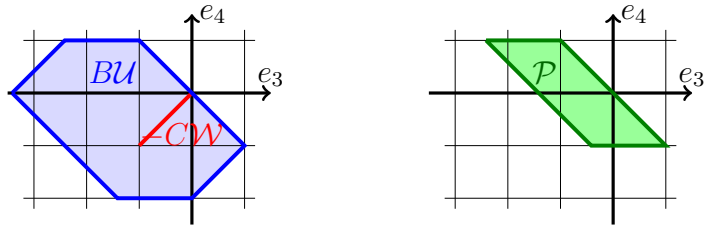
\begin{figure}[htbp!]
    \centering
    \begin{tikzpicture}[scale = 0.7]
        \fill[blue!15!white] (-3.414, 0) -- (-1.414, -2) -- (0, -2) -- (1, -1) -- (-1, 1) -- (-2.414, 1) -- (-3.414, 0);
        
        \draw[ultra thin] (-3.2, 1) -- (1.2, 1);
        \draw[very thick, ->] (-3.5, 0) -- (1.5, 0);
        \node at (1.5, 0.3) {$e_3$};
        \draw[ultra thin] (-3.2, -1) -- (1.2, -1);
        \draw[ultra thin] (-3.2, -2) -- (1.2, -2);
        \draw[ultra thin] (-3, 1.2) -- (-3, -2.2);
        \draw[ultra thin] (-2, 1.2) -- (-2, -2.2);
        \draw[ultra thin] (-1, 1.2) -- (-1, -2.2);
        \draw[very thick, <-] (0, 1.5) -- (0, -2.5);
        \node at (0.4, 1.5) {$e_4$};
        \draw[ultra thin] (1, 1.2) -- (1, -2.2);
    
        \draw[red, line width=1.5pt] (0, 0) --  (-1, -1);
        \node at (-0.2, -0.8) {\textcolor{red}{$-C\mathcal{W}$}};
        \draw[draw=blue, line width=1.5pt] (-3.414, 0) -- (-1.414, -2) -- (0, -2) -- (1, -1) -- (-1, 1) -- (-2.414, 1) -- (-3.414, 0);
        \node at (-1.5, 0.4) {\textcolor{blue}{$B\mathcal{U}$}};
        
         \fill[green!40!white] (8-2.414, 1) -- (8-0.414, -1) -- (9, -1) -- (7, 1)  -- (8-2.414, 1);
        
        \draw[ultra thin] (4.8, 1) -- (9.2, 1);
        \draw[very thick, ->] (4.5, 0) -- (9.5, 0);
        \node at (9.5, 0.3) {$e_3$};
        \draw[ultra thin] (4.8, -1) -- (9.2, -1);
        \draw[ultra thin] (4.8, -2) -- (9.2, -2);
        \draw[ultra thin] (5, 1.2) -- (5, -2.2);
        \draw[ultra thin] (6, 1.2) -- (6, -2.2);
        \draw[ultra thin] (7, 1.2) -- (7, -2.2);
        \draw[very thick, <-] (8, 1.5) -- (8, -2.5);
        \node at (8.4, 1.5) {$e_4$};
        \draw[ultra thin] (9, 1.2) -- (9, -2.2);
        
        \draw[green!50!black, line width = 1.5pt] (8-2.414, 1) -- (8-0.414, -1) -- (9, -1) -- (7, 1)  -- (8-2.414, 1);
        \node at (6.7, 0.4) {\textcolor{green!40!black}{$\mathcal{P}$}};
    \end{tikzpicture}
    \caption{2D projection of $B\mathcal{U}$ (blue), $-C\mathcal{W}$ (red), and their Minkowski difference $\mathcal{P}$ (green) for the case of the loss of control authority over thruster no. 1.}
    \label{fig:sets_BU_CW_1}
\end{figure}

Note that $(0,0)$ is on the boundary of $\mathcal{P}$, so that $\rho_{max} = 0$, no ball of positive radius centered at $0$ can fit inside $\mathcal{P}$. This issue is much more problematic than just preventing us from reusing the proof of Proposition~\ref{prop:resilience}.
Indeed, let $\mathcal{T}_\text{ref} := \big\{ X_\text{ref}(t) : \dot X_\text{ref}(T) = A X_\text{ref} + r R_\theta(t) p_\text{ref}(t),\ t \geq 0 \big\}$ be the reference trajectory of Fig.~\ref{fig:ref orbit}, where control law $p_\text{ref} \in \mathcal{F}(\mathcal{P}_\text{ref})$ is obtained with the trajectory propagation algorithm of \citep{Ortolano}. To produce this trajectory, we need $0 \in \interior(\mathcal{P}_\text{ref})$ as shown on Fig.~\ref{fig: p_ref}.
Since $0 \notin \interior(\mathcal{P})$, we have $\mathcal{P}_\text{ref} \not\subseteq \mathcal{P}$. Therefore, the spacecraft cannot track $\mathcal{T}_\text{ref}$ after the loss of control authority over a thruster other than no. 4.

\begin{figure}[htbp!]
    \centering
    \includegraphics[scale = 0.5]{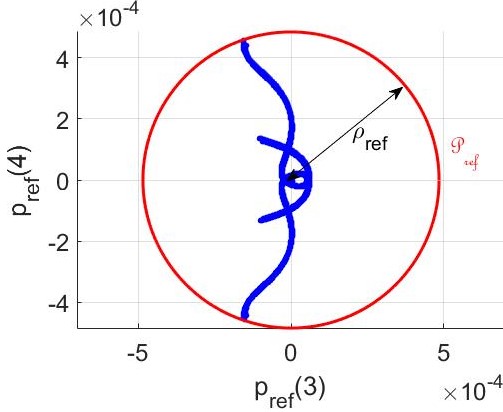}
    \caption{Dimensions 3 and 4 of the reference thrust inputs $p_\text{ref}$ (blue) included in the ball $\mathcal{P}_\text{ref}$ (red) of radius $\rho_\text{ref}$ (black) for the reference trajectory $\mathcal{T}_\text{ref}$.}
    \label{fig: p_ref}
\end{figure}

To simplify further discussions, let us assume that $\mathcal{P}_\text{ref}$ is the smallest ball centered at $0$ encompassing all $p_\text{ref}(t)$, i.e., $\mathcal{P}_\text{ref} = \mathbb{B}^2(0, \rho_\text{ref})$ with $\rho_\text{ref} := \min \big\{\rho > 0 : p_\text{ref}(t) \in \mathbb{B}(0, \rho) \, \text{for all}\ t \geq 0 \big\}$. For the reference trajectory of Fig.~\ref{fig:ref orbit}, radius $\rho_\text{ref} = 4.85 \times 10^{-4}$ and is illustrated on Fig.~\ref{fig: p_ref}.

\subsection{Resilient trajectory tracking and robustness to initial state difference}\label{subsec: tracking no delay}

Following the discussion above, we will only investigate resilient trajectory tracking for the loss of control authority over thruster no. 4. In this scenario, $\rho_{max} = 0.414 >> \rho_\text{ref} = 4.85 \times 10^{-4}$. Then, the malfunctioning spacecraft has a large amount of control authority left even after counteracting the worst undesirable thrust and producing the reference thrust input. Let us detail why this remaining thrust capability will be sorely needed.

The initial state of the malfunctioning spacecraft $X_0$ is most likely not exactly equal to $X_\text{ref}(0)$, the initial of reference trajectory $\mathcal{T}_\text{ref}$, which was designed before the spaceflight. We then need to design a tracking controller with robustness to uncertainty on the initial state. Moreover, if the difference $X_0 - X_\text{ref}(0)$ is not actively reduced, it can grow exponentially with time \citep{Khalil}. Thus, we need the extra thrust capability mentioned earlier to counteract $X(t) - X_\text{ref}(t)$.
Formally, we pick $\varepsilon > 0$ and define input set $\mathcal{P}_\varepsilon := \mathbb{B}^2(0, \varepsilon)$ to overcome $X(t) - X_\text{ref}(t)$. For the robust tracking of $\mathcal{T}_\text{ref}$ to be admissible, we then need $\mathcal{P}_\varepsilon \oplus \mathcal{P}_\text{ref} \subseteq \mathcal{P}$, where we recall $\mathcal{P}$ as the set of control inputs remaining after counteracting the worst undesirable thrust from malfunctioning thruster no. 4. We now introduce the dynamics tasked with counteracting the initial state error
\begin{equation}\label{eq: Hajek epsilon rotated}
    \dot Y(t) = AY(t) + r R_\theta(t) p_\varepsilon(t), \quad Y(0) = X_0 - X_\text{ref}(0), \quad p_\varepsilon(t) \in \mathcal{P}_\varepsilon,
\end{equation}
where $R_\theta(t)$ is the rotation matrix tracking position $X(t)$ of system~\eqref{eq:split ODE rotation}.
\begin{proposition}\label{prop: res traj track rotated}
    If $\varepsilon + \rho_\text{ref} \leq \rho_{max}$, then system~\eqref{eq: Hajek epsilon rotated} is stabilizable in a finite time $t_f$ and the reference trajectory $\mathcal{T}_\text{ref}$ can be tracked exactly by system~\eqref{eq:split ODE rotation} after time $t_f$, i.e., $X(T) = X_\text{ref}(T)$ for all $t \geq t_f$.
\end{proposition}
\begin{proof}
    We start with the same trick as in the proof of Proposition~\ref{prop:resilience} by noticing that $\mathcal{P}_\varepsilon = \mathbb{B}^2(0, \varepsilon)$ is left unchanged by the rotation matrix $R_\theta$. Then, system~\eqref{eq: Hajek epsilon rotated} has a one-to-one correspondence with the following linear system
    \begin{equation}\label{eq: ODE Bepsilon}
        \dot Y(t) = AY(t) + r \hat{B} p_\varepsilon(t), \quad Y(0) = X_0 - X_\text{ref}(0), \quad p_\varepsilon(t) \in \mathcal{P}_\varepsilon, \quad \hat{B} = \left[\begin{smallmatrix} 0_{2 \times 2} \\ I_2 \end{smallmatrix}\right].
    \end{equation}
    Since $0 \in \interior(\mathcal{P}_\varepsilon)$, $Re(\Lambda(A))\leq 0$, and $\rank\big( \big[ \hat{B}\, A\hat{B} \big] \big) = 4$, Corollary 3.6 of \citep{Brammer} states that system~\eqref{eq: ODE Bepsilon} is stabilizable in a finite time $t_f$ and so is system~\eqref{eq: Hajek epsilon rotated} by construction.
    
    Therefore, there exists a signal $p_\varepsilon \in \mathcal{F}(\mathcal{P}_\varepsilon)$ on $[0, t_f]$ yielding $Y(t_f) = 0$ in system~\eqref{eq: Hajek epsilon rotated}. Since $0 \in \mathcal{P}_\varepsilon$, we extend the control signal to $p_\varepsilon(t) = 0$ for all $t > t_f$. We now define the control law $p_{track}(t) := p_\varepsilon(t) + p_\text{ref}(t)$. Note that $\mathcal{P}_\varepsilon \oplus \mathcal{P}_\text{ref} = \mathbb{B}(0, \varepsilon) \oplus \mathbb{B}(0, \rho_\text{ref}) = \mathbb{B}(0, \varepsilon + \rho_\text{ref}) \subseteq \mathbb{B}(0, \rho_{max})$ since $\varepsilon + \rho_\text{ref} \leq \rho_{max}$. By definition of $\rho_{max}$, $\mathbb{B}(0, \rho_{max}) \subseteq \mathcal{P}$. Thus,  $\mathcal{P}_\varepsilon \oplus \mathcal{P}_\text{ref} \subseteq \mathcal{P}$, i.e., $p_{track}(t) \in \mathcal{P}$ for all $t \geq 0$.
    
    Let $w \in \mathcal{F}(\mathcal{W})$ be any undesirable input signal. Then, by definition of $\mathcal{P}$, there exists $u \in \mathcal{F}(\mathcal{U})$ such that $Bu(t) = p_{track}(t) - Cw(t)$ for all $t \geq 0$. We now implement this controller for $T \geq t_f$ in system~\eqref{eq:split ODE rotation}:
    \begin{align*}
        X(T) &= e^{AT}\left(X_0 + \int_{0}^T \hspace{-3mm} e^{-At} r R_\theta(t) \big( Bu(t) + Cw(t) \big)\, dt \right) \\
        &= e^{AT}\left(X_0 + \int_{0}^T \hspace{-3mm} e^{-At} r R_\theta(t) \big( p_\varepsilon(t) + p_\text{ref}(t) \big)\, dt \right)\\
        &= e^{AT}\left(X_0 + \int_{0}^{T} \hspace{-3mm} e^{-At} r R_\theta(t) p_\varepsilon(t)\, dt + e^{-AT} X_\text{ref}(T) - X_\text{ref}(0) \right),
    \end{align*}
    because $X_\text{ref}(T) = e^{AT}\left( X_\text{ref}(0) + \int_0^T e^{-At} r R_\theta(t) p_\text{ref}(t)dt \right)$.
    Then, 
    \begin{align*}
        X(T) - X_\text{ref}(T) &= e^{AT}\left( X_0 - X_\text{ref}(0) + \int_0^{T} \hspace{-3mm} e^{-At} r R_\theta(t) p_\varepsilon(t) \, dt \right) \\
        &= e^{AT}\left( X_0 - X_\text{ref}(0) + \int_0^{t_f} \hspace{-3mm} e^{-At} r R_\theta(t) p_\varepsilon(t) \, dt \right),
    \end{align*}
    since $p_\varepsilon(t) = 0$ for $t > t_f$. By definition of $p_\varepsilon$,
    \begin{equation*}
        Y(t_f) = 0 = e^{At_f} \left( Y(0) + \int_0^{t_f} \hspace{-3mm} e^{-At} r R_\theta(t) p_\varepsilon(t)\, dt \right), \quad \text{i.e.}, \quad X_0 - X_\text{ref}(0) + \int_0^{t_f} \hspace{-3mm} e^{-At} r R_\theta(t) p_\varepsilon(t)\, dt = 0.
    \end{equation*}
   Therefore, $X(T) = X_\text{ref}(T)$ for all $T \geq t_f$.
\end{proof}

Proposition~\ref{prop: res traj track rotated} states that as long as $\varepsilon + \rho_\text{ref} \leq \rho_{max}$, there exists a finite time $t_f$ after which any trajectory $\mathcal{T}_\text{ref}$ can be tracked perfectly despite the loss of control authority over a thruster.
Since $\varepsilon$ describes the maximal input magnitude of system~\eqref{eq: Hajek epsilon rotated}, $\varepsilon$ is inversely correlated with its stabilization time $t_f$. Then, the constraint $\varepsilon + \rho_\text{ref} \leq \rho_{max}$ yields that the smaller $\rho_\text{ref}$, the larger $\varepsilon$ and so the smaller $t_f$ is. In other words, the smaller the inputs required to track the reference trajectory, the faster the spacecraft can resume perfect tracking after a loss of control authority.
Let us now investigate how the spacecraft would perform if the controller could not react instantly to undesirable thrust inputs.

\section{Resilience theory in the presence of actuation delay}\label{sec:theory delay}

In order to account for the unavoidable sensors and thrusters delays on spacecraft \citep{thruster_delay}, we now assume that the controller operates with a constant input delay $\tau > 0$ so that the dynamics of the spacecraft are in fact
\begin{equation}\label{eq:delayed split ODE rotation}
    \dot X(t) = A X(t) + r R_\theta(t) Bu\big(t, X(t-\tau), w(t-\tau) \big) + r R_\theta(t) C w(t), \quad X(0) = X_0 \in \mathbb{R}^4.
\end{equation}
The controller cannot react immediately to a change of the undesirable input $w(t)$ and cancel it instantaneously as in Section~\ref{sec:theory NO delay}. Only starting at $t+\tau$ can the controller try to counteract $Cw(t)$. However, at time $t+\tau$ the effect of $w(t)$ on the state $X(t+\tau)$ has become $e^{A\tau}Cw(t)$. Hence, set $\mathcal{P}$ introduced in \eqref{eq:Hajek ODE} does not describe the remaining control authority anymore.

To study whether the more realistic spacecraft dynamics \eqref{eq:delayed split ODE rotation} can resiliently track a reference trajectory, we first need to establish several novel theoretical results. Indeed, resilience in the presence of actuation delay has never been investigated before. We first want to establish general results, before specializing them to dynamics \eqref{eq:delayed split ODE rotation}. Let us first work on general linear systems with input delay and no rotation $R_\theta$. The rotation matrix would only obfuscate the already complex theory of resilience in the presence of actuation delay.

\subsection{Framework for actuation delay}\label{subsec: delay res frame}

We study a linear system of initial dynamics $\dot x(t) = Ax(t) + \bar{B}\bar{u}(t)$ with $x(0) = x_0 \in \mathbb{R}^n$, $\bar{u}(t) \in \bar{\mathcal{U}}$ an hyperrectangle in $\mathbb{R}^{m+q}$ and $A \in \mathbb{R}^{n \times n}$, $\bar{B} \in \mathbb{R}^{n \times (m+q)}$ constant matrices. Similarly to the motivating spacecraft scenario, this system suffers a loss of control authority over $q$ of its initial $m+q$ actuators, and the controller is further inflicted with a constant actuation delay $\tau > 0$. We split the signal $\bar{u}$ into its controlled part $u \in \mathcal{F}(\mathcal{U})$ and its uncontrolled part $w \in \mathcal{F}(\mathcal{W})$ with $\mathcal{U}$ and $\mathcal{W}$ hyperrectangles in $\mathbb{R}^m$ and $\mathbb{R}^q$ respectively. Matrix $\bar{B}$ is accordingly split into $B \in \mathbb{R}^{n \times m}$ and $C \in \mathbb{R}^{n \times q}$ such that the dynamics of the malfunctioning system are
\begin{equation}\label{eq: split delayed}
    \dot x(t) = Ax(t) + C w(t) + B u\big(t, x(t-\tau), w(t-\tau)\big), \quad x(0) = x_0 \in \mathbb{R}^n, \quad u(t) \in \mathcal{U}, \quad w(t) \in \mathcal{W}.
\end{equation}
We want to establish conditions under which system~\eqref{eq: split delayed} can resiliently reach a target set.

\begin{definition}\label{def: linear res reach}
    A convex set $\mathcal{G} \subseteq \mathbb{R}^n$ is \emph{resiliently reachable} at time $T$ from $x_0$ by system~\eqref{eq: split delayed} if for all undesirable inputs $w \in \mathcal{F}(\mathcal{W})$ there exists a control signal $u \in \mathcal{F}(\mathcal{U})$ such that $u(t) = u\big(t, x(t-\tau), w(t-\tau)\big)$ and $x(T) \in \mathcal{G}$.
\end{definition}

Generalizing works \citep{Hajek, ECC} we introduce the family of sets $\mathcal{P}_t := B\mathcal{U} \ominus (- e^{At} C\mathcal{W})$ for all $t \geq 0$ with $e^{A t} C\mathcal{W} := \big\{ e^{A t} Cw : w \in \mathcal{W} \big\}$. We will show that $\mathcal{P}_\tau$ is the set of actual control inputs of system~\eqref{eq: split delayed}, when $u$ has canceled any undesirable input $w$ with a delay $\tau$. Set $\mathcal{P}_\tau$ is the time delayed extension of $\mathcal{P}$ from \eqref{eq:Hajek ODE}.

\begin{definition}\label{def: T_c}
    The \emph{minimal correction time} $T_c$ represents the minimal time after which any undesirable input can be counteracted, $T_c := \inf\big\{ t \geq \tau : \mathcal{P}_t \neq \emptyset \big\}$.
\end{definition}

If $T_c = +\infty$, then the impact on the state of some undesirable inputs cannot be canceled by any control input after the actuation delay, i.e., there exists some $w \in \mathcal{W}$ such that $-e^{At}Cw \notin B\mathcal{U}$ for all $t \geq \tau$. In this case, resilient reachability is impossible \citep{ECC_extended}. Let us now assume that $T_c$ is finite in order to build the theory for resilience in the presence of actuation delay.

\subsection{Resilient reachability despite actuation delay}\label{subsec: delayed res reach}

We want to know whether a target set $\mathcal{G} \subseteq \mathbb{R}^n$ is resiliently reachable by system~\eqref{eq: split delayed}. Because of the actuation delay $\tau$, the controller can only guarantee that $x(t)$ is in some neighborhood of $x(t-\tau)$, it cannot ensure an exact location. Then, set $\mathcal{G}$ needs a minimal radius $\rho > 0$ to be resiliently reachable. 
Inspired by H\'ajek's approach \citep{Hajek}, we introduce system~\eqref{eq: Hajek} as a counterpart to system~\eqref{eq: split delayed}, just like system~\eqref{eq:Hajek ODE} was the counterpart of system~\eqref{eq:split ODE rotation} for the spacecraft without actuation delay. \vspace{-3mm}
\begin{equation}\label{eq: Hajek}
    \dot x(t) = Ax(t) + p(t), \quad x(0) = e^{A T_c} x_0, \quad p(t) \in \mathcal{P}_{T_c} := B\mathcal{U} \ominus (- e^{A T_c} C\mathcal{W}).
\end{equation}
Note that the input $p$ of system~\eqref{eq: Hajek} does not suffer from actuation delay by design of $\mathcal{P}_{T_c}$.

\begin{theorem}\label{thm: delayed res reach}
     If there exists $x_g \in \mathcal{G}$ such that $\mathbb{B}(x_g, \rho) \subseteq \mathcal{G}$ and $x_g$ is reachable in a finite time $T$ by system~\eqref{eq: Hajek}, then $\mathcal{G}$ is resiliently reachable by system~\eqref{eq: split delayed} in time $T + T_c$, with $\rho := \frac{c}{\mu(A)}\big( e^{\mu(A)T_c} - 1 \big)$ and $c := \max\big\{ \|Cw\| : w \in \mathcal{W} \big\}$.
\end{theorem}
\begin{proof}
    We first introduce the log-norm of matrix $A$ defined in \citep{exp} as $\mu(A) := \max \big\{ \Lambda((A+A^\top)/2) \big\}$. Then, $\|e^{At}\| \leq e^{\mu(A)t}$ for all $t \geq 0$ according to. 
    Since $T$ is the time at which system~\eqref{eq: Hajek} can reach $x_g$ from $e^{A T_c} x_0$, there exists $p(s) \in \mathcal{P}_{T_c}$ for all $s \in [0, T]$ such that 
    \begin{equation*}
        x_g = x(T) = e^{AT} \left( e^{AT_c} x_0 + \int_0^{T} \hspace{-2mm} e^{-As} p(s)\, ds \right), \quad \text{i.e.,} \quad e^{AT} \int_0^{T} \hspace{-2mm} e^{-As} p(s)\, ds = x_g - e^{A(T+T_c)}x_0.
    \end{equation*}
    
    Let $w \in \mathcal{F}(\mathcal{W})$ be some undesirable input affecting system~\eqref{eq: split delayed}. We now define the corresponding control input $u \in \mathcal{F}(\mathcal{U})$ so that it satisfies: $Bu(t) = 0$ for $t \in [0, T_c]$ and $Bu(t) = p(t-T_c) - e^{AT_c}Cw(t-T_c)$ for $t \in [T_c,\, T + T_c]$. Note that $u(t) \in \mathcal{U}$ by definition of $p(t) \in \mathcal{P}_{T_c}$. We apply this control law to system~\eqref{eq: split delayed}:
    \begin{align*}
        x(T+T_c) &= e^{A(T+T_c)}\left( x_0 + \int_0^{T+T_c} \hspace{-2mm} e^{-At} C w(t)\, dt + \int_{0}^{T+T_c} \hspace{-2mm} e^{-At} Bu(t) dt \right)\\
        &= e^{A(T+T_c)}\left( x_0 + \int_0^{T+T_c} \hspace{-2mm} e^{-At} C w(t)\, dt + \int_{T_c}^{T+T_c} \hspace{-2mm} e^{-At} \big( p(t-T_c) - e^{AT_c} Cw(t-T_c)\big) dt \right)\\
        &= e^{A(T+T_c)}\left( x_0 + \int_0^{T+T_c} \hspace{-2mm} e^{-At} C w(t)\, dt + \int_{0}^{T} \hspace{-2mm} e^{-A(s+T_c)} \big( p(s) - e^{AT_c}Cw(s)\big) ds \right)\\
        &= e^{A(T+T_c)}\left( x_0 + \int_{T}^{T+T_c} \hspace{-2mm} e^{-At} Cw(t) dt \right) + e^{AT}\int_0^{T} \hspace{-2mm} e^{-As} p(s) ds \\
        &= e^{A(T+T_c)} x_0 + \int_{0}^{T_c} \hspace{-2mm} e^{As} Cw(T+T_c-s) ds + \big( x_g - e^{A(T+T_c)}x_0 \big) \\
        &= x_g + \int_{0}^{T_c} \hspace{-2mm} e^{As} Cw(T+T_c-s) ds,
    \end{align*}
    thanks to the definition of $p(s) \in \mathcal{P}_{T_c}$. Then, by subtracting $x_g$ and using the triangle inequality we obtain
    \begin{equation*}
        \big\|x(T+T_c) - x_g \big\| \leq \int_{0}^{T_c} \hspace{-1mm} \big\|e^{As}\big\|\, \big\|Cw(T+T_c-s)\big\| ds \leq \int_{0}^{T_c} \hspace{-2mm} e^{\mu(A)s} c\,ds = \frac{c}{\mu(A)}\big( e^{\mu(A)T_c} - 1 \big) = \rho.
    \end{equation*}
    Since $\mathbb{B}(x_g, \rho) \subseteq \mathcal{G}$, we have $x(T+T_c) \in \mathcal{G}$. Hence, $\mathcal{G}$ is resiliently reachable by system~\eqref{eq: split delayed} in time $T+T_c$.
\end{proof}

Note that the control $p(t-T_c)$ responsible for steering to $x_g$ in Theorem~\ref{thm: delayed res reach} is in fact an \textit{open loop} control. A \textit{feedback control} may perform better in practice, but the saturation enforcing that such signal remains bounded in $\mathcal{P}_{T_c}$ leads to a substantial increase in complexity, as demonstrated in Section~\ref{subsec: feedback}.

Theorem~\ref{thm: delayed res reach} provides a sufficient resilient reachability condition for delayed system~\eqref{eq: split delayed} in terms of the reachability of $\mathcal{G}$ by system~\eqref{eq: Hajek}. In turn, a sufficient condition for this last property can be verified with the lemma below.

\begin{lemma}\label{lemma: controllability}
    Let $P \in \mathbb{R}^{n \times d}$ be a matrix whose columns are $d$ linearly independent vectors of $\mathcal{P}_{T_c}$ with $d := \dim(\mathcal{P}_{T_c}) \geq 0$.
    If $0 \in \interior(\mathcal{P}_{T_c})$, then system~\eqref{eq: Hajek} is controllable if and only if $Re(\Lambda(A)) = 0$, $\rank(\mathcal{C}(A,P)) = n$, and there is no real eigenvector $v$ of $A^\top$ satisfying $v^\top p \leq 0$ for all $p \in \mathcal{P}_{T_c}$.
\end{lemma}
\begin{proof}
    By construction, $\Image(P) = \Span(\mathcal{P}_{T_c})$ so that $\mathcal{C}(A,P)$ is a controllability matrix associated with system~\eqref{eq: Hajek}.
    Since $\mathcal{U}$ and $\mathcal{W}$ are convex, so are $B\mathcal{U}$ and $e^{A T_c} C\mathcal{W}$. Their Minkowski difference $\mathcal{P}_{T_c}$ is then also convex \citep{Pontryagin_difference}. Then, $0 \in \interior(\co(\mathcal{P}_{T_c}))$ and trivially $0 \in \Ker(I_n) \cap \mathcal{P}_{T_c}$. These two inclusions allow us to apply Corollary 3.7 of \citep{Brammer} which yields the controllability condition.
\end{proof}

Thus, combining Lemma~\ref{lemma: controllability} and Theorem~\ref{thm: delayed res reach} provides a sufficient condition for resilient reachability in the presence of actuation delay. We will now investigate the more complicated problem of resilient trajectory tracking despite actuation delay.

\subsection{Resilient trajectory tracking despite actuation delay}\label{subsec: delayed res traj tracking}

We want system~\eqref{eq: split delayed} to track the actuated reference trajectory $\mathcal{T}_\text{ref}$ designed for system~\eqref{eq: Hajek} by $\mathcal{T}_\text{ref} := \big\{ x_\text{ref}(t) : \dot x_\text{ref}(t) = A x_\text{ref}(t) + p_\text{ref}(t), \ \text{for all}\ t \geq 0 \big\}$ with $p_\text{ref} \in \mathcal{F}(\mathcal{P}_\text{ref})$. We also define $\mathcal{P}_\varepsilon$ as a compact set of $\mathbb{R}^n$ satisfying $0 \in \relint(\mathcal{P}_\varepsilon)$ and $\dim(\mathcal{P}_\varepsilon) = n$. As in Section~\ref{subsec: tracking no delay}, we use input set $\mathcal{P}_\varepsilon$ to counteract the error arising from $x_0 \neq x_\text{ref}(0)$ through the dynamics
\begin{equation}\label{eq: Hajek epsilon}
    \dot y(t) = Ay(t) + p_\varepsilon(t), \quad y(0) = e^{AT_c}\big( x_0 - x_\text{ref}(0) \big), \quad p_\varepsilon(t) \in \mathcal{P}_\varepsilon.
\end{equation}
We can then state our resilient trajectory tracking result.
\begin{theorem}\label{thm: delayed res traj track}
    If $\mathcal{P}_\varepsilon \oplus \mathcal{P}_\text{ref} \subseteq \mathcal{P}_{T_c}$ and system~\eqref{eq: Hajek epsilon} is stabilizable in a finite time $t_f$, then the reference trajectory $\mathcal{T}_\text{ref}$ can be tracked by system~\eqref{eq: split delayed} with a precision $\rho$ after time $t_f + T_c$, i.e., $\|x(T) - x_\text{ref}(T)\| \leq \rho$ for all $T \geq t_f + T_c$.
\end{theorem}
\begin{proof}
    Since system~\eqref{eq: Hajek epsilon} is stabilizable in a finite time $t_f$, there exists a signal $p_\varepsilon(t) \in \mathcal{P}_\varepsilon$ for all $t \in [0, t_f]$ yielding $y(t_f) = 0$. Because $0 \in \mathcal{P}_\varepsilon$, we can extend signal $p_\varepsilon$ with $p_\varepsilon(t) = 0$ for all $t > t_f$. 
    Let $w \in \mathcal{F}(\mathcal{W})$ be the undesirable input signal. Since $\mathcal{P}_\varepsilon \oplus \mathcal{P}_\text{ref} \subseteq \mathcal{P}_{T_c} = B\mathcal{U} \ominus (-e^{AT_c}C\mathcal{W})$, there exists $u \in \mathcal{F}(\mathcal{U})$ such that $Bu(t) = p_\text{ref}(t) + p_\varepsilon(t-T_c) - e^{A T_c} Cw(t-T_c)$ for $t \geq T_c$.
    Indeed, the reference trajectory is known ahead of time, so $u(t)$ has access to $p_\text{ref}(t)$. We define $u$ to satisfy $Bu(t) = p_\text{ref}(t)$ for $t \in [0, T_c]$.
    We now implement this controller for $T \geq T_c$ in system~\eqref{eq: split delayed}
    \begin{align*}
        x(T) &= e^{AT}\left(x_0 + \int_0^{T_c} \hspace{-3mm} e^{-At} Bu(t)\, dt + \int_{T_c}^T \hspace{-3mm} e^{-At} Bu(t)\, dt +  \int_0^T \hspace{-3mm} e^{-At} Cw(t)\, dt \right) \\
        &= e^{AT} \hspace{-1mm} \left( \hspace{-1mm} x_0 \hspace{-1mm} + \hspace{-1mm} \int_0^{T_c} \hspace{-4mm} e^{-At} p_\text{ref}(t) dt + \hspace{-1mm} \int_{T_c}^T \hspace{-3mm} e^{-At} \big( p_\text{ref}(t) + p_\varepsilon(t-T_c) - e^{AT_c}Cw(t-T_c) \big) dt + \hspace{-1mm} \int_0^T \hspace{-3mm} e^{-At} Cw(t) dt \hspace{-1mm}  \right) \\
        &= e^{AT}\left(x_0 + \int_0^T \hspace{-3mm} e^{-At} p_\text{ref}(t)\, dt + \int_0^{T-T_c} \hspace{-3mm} e^{-As}\big(e^{-A T_c}p_\varepsilon(s) - Cw(s)\big) ds +  \int_0^T \hspace{-3mm} e^{-At} Cw(t)\, dt \right) \\
        &= e^{AT}\left(x_0 + e^{-AT} x_\text{ref}(T) - x_\text{ref}(0) + e^{-A T_c} \int_0^{T-T_c} \hspace{-3mm} e^{-As} p_\varepsilon(s)\, ds +  \int_{T-T_c}^T \hspace{-3mm} e^{-At} Cw(t)\, dt \right),
    \end{align*}
    by definition of $p_\text{ref}$. Then,
    \begin{equation}\label{eq: tracking}
        x(T) - x_\text{ref}(T) = e^{A(T-T_c)}\left( e^{AT_c}\big(x_0 - x_\text{ref}(0)\big) + \int_0^{T-T_c} \hspace{-3mm} e^{-As} p_\varepsilon(s)\, ds \right) +  \int_{0}^{T_c} \hspace{-3mm} e^{As} Cw(T-s)\, ds.
    \end{equation}
    Note that the last integral term is the same as in Theorem~\ref{thm: delayed res reach} and hence can be bounded similarly:
    \begin{equation*}
        \left\| \int_{0}^{T_c} \hspace{-3mm} e^{As} Cw(T-s)\, ds \right\| \leq \int_{0}^{T_c} \hspace{-1mm} \big\|e^{As} \big\|\, \big\|Cw(T-s) \big\| \, ds \leq c \int_0^{T_c} \hspace{-3mm} e^{\mu(A)s}ds = \rho.
    \end{equation*}
    Since $p_\varepsilon$ stabilizes system~\eqref{eq: Hajek epsilon} in time $t_f$, we have $y(T) = 0$ for all $T \geq t_f$. In particular, for $T \geq t_f + T_c$ we obtain $y(T-T_c) = 0 = e^{A(T-T_c)} \left( y(0) + \int_0^{T-T_c} e^{-As} p_\varepsilon(s)\, ds \right)$. Note that $y(T-T_c)$ is exactly the central term in \eqref{eq: tracking}, which finally yields $\big\| x(T) - x_\text{ref}(T) \big\| \leq \rho$.
\end{proof}

Without control signal $p_\varepsilon$, the tracking error would be $\|x(t) - x_\text{ref}(t)\| \leq \rho + \big\|e^{At}\big( x_0 - x_\text{ref}(0) \big)\big\|$, which can grow exponentially if $x_0 - x_\text{ref}(0)$ is collinear with a positive eigenvector of $A$. When $x_0 = x_\text{ref}(0)$, we do not need $p_\varepsilon$ and the tracking can be performed with precision $\rho$ from time $T_c$ onward. 

Theorem~\ref{thm: delayed res traj track} provides a sufficient condition for resilient trajectory tracking by delayed system~\eqref{eq: split delayed} in terms of the finite time stabilizability of system~\eqref{eq: Hajek epsilon}. In turn, this property can be verified with the lemma below.

\begin{lemma}\label{lemma: stabilizability}
    Let $P_\varepsilon \in \mathbb{R}^{n \times d}$ be a matrix whose columns are $d$ linearly independent vectors of $\mathcal{P}_\varepsilon$, with $d = \dim(\mathcal{P}_{T_c}) = \dim(\mathcal{P}_\varepsilon)$.
    System~\eqref{eq: Hajek epsilon} is stabilizable in a finite time if and only if $Re(\Lambda(A))\leq 0$, $\rank(\mathcal{C}(A,P_\varepsilon)) = n$ and there is no real eigenvector $v$ of $A^\top$ satisfying $v^\top p \leq 0$ for all $p \in \mathcal{P}_\varepsilon$.
\end{lemma}
\begin{proof}
    By construction $\Image(P_\varepsilon) = \Span(\mathcal{P}_\varepsilon)$, so that $\mathcal{C}(A,P_\varepsilon)$ is a controllability matrix associated with system~\eqref{eq: Hajek epsilon}. We made the assumption that $0 \in \relint(\mathcal{P}_\varepsilon)$, so we can apply Proposition~2 of \citep{ECC_extended} which guarantees that system~\eqref{eq: Hajek epsilon} is stabilizable in finite time.
\end{proof}

Now that we have established sufficient conditions for resilient reachability and resilient trajectory tracking for linear systems in the presence of actuation delay, we will investigate how to adapt this theory to the spacecraft dynamics \eqref{eq:delayed split ODE rotation}.

\section{Spacecraft resilience in the presence of actuation delay}\label{sec:application}

In this section we extend the linear theory of Section~\ref{sec:theory delay} to the rotating dynamics~\eqref{eq:delayed split ODE rotation} to build an answer to Problem~\ref{prob:mission}. We start by verifying whether the open-loop controller of Section~\ref{subsec: delayed res traj tracking} can be applied to the spacecraft.

\subsection{Open-loop controller}\label{subsec: T_c too large}

To adapt the resilient trajectory tracking controller of Theorem~\ref{thm: delayed res traj track} to nonlinear system~\eqref{eq:delayed split ODE rotation}, we need to modify the minimal correction time $T_c$ introduced for linear systems in Definition~\ref{def: T_c}. However, we will show that such a modification is in fact impossible and prevents a straightforward extension of Theorem~\ref{thm: delayed res traj track} to nonlinear system~\eqref{eq:delayed split ODE rotation}.
If we were to adapt $T_c$ to system~\eqref{eq:delayed split ODE rotation}, we need to understand the reasoning behind Definition~\ref{def: T_c}, which comes from the calculations of Theorem~\ref{thm: delayed res traj track}, where $T_c$ allows $Bu(T)$ to counteract the effect of $w(T-T_c)$ on the state $x(T)$. We will then emulate the proof of Theorem~\ref{thm: delayed res traj track} for nonlinear dynamics~\eqref{eq:delayed split ODE rotation} in the vain hope of determining an updated $T_c$.

The reference trajectory is $\mathcal{T}_\text{ref} = \big\{X_\text{ref}(t) : \dot X_\text{ref}(t) = A X_\text{ref}(t) + r p_\text{ref}(t)$ for all $t \geq 0 \big\}$, $p_\text{ref} \in \mathcal{F}(\mathcal{P}_\text{ref})$ with $\mathcal{P}_\text{ref} = \mathbb{B}^2(0, \rho_\text{ref})$.
Inspired by Theorem~\ref{thm: delayed res traj track}, for $w \in \mathcal{F}(\mathcal{W})$ we define the control law $Bu(t) = R_\theta^{-1}(t) p_\text{ref}(t)$ for $t \in [0, T_c]$ and $Bu(t) = R_\theta^{-1}(t) p_\text{ref}(t) + p_\varepsilon(t-T_c) + Bu_w(t-T_c)$ for $t \geq T_c$, where $p_\varepsilon$ should counteract the error arising from $X_0 - X_\text{ref}(0)$ and $Bu_w$ should cancel the effect of $w$.
We apply this control law to the spacecraft dynamics~\eqref{eq:delayed split ODE rotation} for $T \geq T_c$:
\begin{align*}
    X(T) &= e^{AT}\left( X_0 + \int_0^T \hspace{-3mm} e^{-At}r R_\theta(t) \big( Bu(t) + Cw(t) \big) dt \right) \\
    &= e^{AT}\hspace{-1mm} \left( \hspace{-1mm} X_0 + \hspace{-1mm} \int_0^T \hspace{-3mm} e^{-At}r p_\text{ref}(t)dt + \hspace{-1mm} \int_{T_c}^T \hspace{-3mm} e^{-At}r R_\theta(t) \big( p_\varepsilon(t-T_c) + Bu_w(t-T_c) \big)dt + \hspace{-1mm} \int_0^T \hspace{-3mm} e^{-At}r R_\theta(t) Cw(t) dt \hspace{-1mm} \right) \\
    &= e^{AT}\left( X_0 + e^{-AT} X_\text{ref}(T) - X_\text{ref}(0) + \int_0^{T-T_c} \hspace{-4mm} e^{-A(t+T_c)}r R_\theta(t+T_c) \big( p_\varepsilon(t) + Bu_w(t) \big)dt \right) \\
    & \quad + \int_0^T \hspace{-3mm} e^{A(T-t)}r R_\theta(t) Cw(t) dt  \\
    &= X_\text{ref}(T) + e^{A(T-T_c)}\left( e^{AT_c} \big(X_0 - X_\text{ref}(0)\big) + \int_0^{T-T_c} \hspace{-3mm} e^{-A(t+T_c)}r R_\theta(t+T_c)p_\varepsilon(t)\right) \\
    & \quad+ \int_{T-T_c}^T \hspace{-3mm} e^{A(T-t)}r R_\theta(t) Cw(t) dt + e^{AT} \int_0^{T-T_c} \hspace{-2mm} e^{-At}r \big[ R_\theta(t)Cw(t) + e^{-AT_c} R_\theta(t+T_c)Bu_w(t) \big] dt.
\end{align*}
The last term in square brackets is the one to be canceled by the appropriate choice of $Bu_w$ and leads to the updated definition
\begin{equation}\label{eq: T_c nonlinear}
    T_c = \inf \big\{ T \geq \tau : -R_\theta^{-1}(t+T)e^{AT}R_\theta(t)C\mathcal{W} \subseteq B\mathcal{U} \ \text{for all}\ t \geq 0 \big\}.
\end{equation}
However, this definition of $T_c$ creates a circular dependency on $R_\theta(t)$, which depends on state $X(t)$, which is in turn modified by controller $u$ which relies on $T_c$. Then, \eqref{eq: T_c nonlinear} properly defines $T_c$ only if it is invariant with respect to $\theta(t)$. Let us investigate this question by calculating $R_\theta^{-1}(t+T)e^{AT}R_\theta(t)C$.

We first calculate
\begin{equation*}
    A = \begin{bmatrix} 0 & 0 & 1 & 0 \\ 0 & 0 & 0 & 1 \\ 3\Omega^2 & 0 & 0 & 2\Omega \\ 0 & 0 & -2\Omega & 0 \end{bmatrix} \hspace{-1mm}, \quad A^2 = \begin{bmatrix} 3\Omega^2 & 0 & 0 & 2\Omega \\ 0 & 0 & -2\Omega & 0 \\ 0 & 0 & -\Omega^2 & 0 \\ -6\Omega^3 & 0 & 0 & -4\Omega^2 \end{bmatrix} \hspace{-1mm}, \quad A^3 = \begin{bmatrix} 0 & 0 & -\Omega^2 & 0 \\ -6\Omega^3 & 0 & 0 & -4\Omega^2 \\ -3\Omega^4 & 0 & 0 & -2\Omega^3 \\ 0 & 0 & 2\Omega^3 & 0 \end{bmatrix} \hspace{-1mm},
\end{equation*}
and to calculate $e^{AT}$, we will establish two simple recursions.
\begin{lemma}\label{lemma: recursion}
    For all $n \in \mathbb{N}$,
    \begin{equation}\label{eq: recursion}
       A^{2n + 2} = (-1)^n \Omega^{2n} A^2 \qquad \text{and} \qquad A^{2n + 3} = (-1)^n \Omega^{2n} A^3.
    \end{equation}
\end{lemma}
\begin{proof}
    For $n = 0$, \eqref{eq: recursion} is trivially valid. Assuming that \eqref{eq: recursion} holds for some $n \in \mathbb{N}$, let us now show that \eqref{eq: recursion} also holds for $n+1$. 

    First, notice that
    \begin{equation*}
        A^4 = \begin{bmatrix} -3\Omega^4 & 0 & 0 & -2\Omega^3 \\ 0 & 0 & 2\Omega^3 & 0 \\ 0 & 0 & \Omega^4 & 0 \\ 6\Omega^5 & 0 & 0 & 4\Omega^4 \end{bmatrix} = -\Omega^2 A^2, \qquad \text{and} \qquad A^5 = \begin{bmatrix} 0 & 0 & \Omega^4 & 0 \\ 6\Omega^5 & 0 & 0 & 4\Omega^4 \\ 3\Omega^6 & 0 & 0 & 2\Omega^5 \\ 0 & 0 & -2\Omega^5 & 0 \end{bmatrix} = -\Omega^2 A^3.
    \end{equation*}
    For $n+1$ we calculate $A^{2(n+1)+2} = A^2 A^{2n+2} = A^2 (-1)^n \Omega^{2n} A^2$ because we assumed that \eqref{eq: recursion} holds at $n$. Then, using $A^4 = -\Omega^2 A^2$, we obtain $A^{2(n+1)+2} = (-1)^{n+1} \Omega^{2(n+1)} A^2$. Similarly, using $A^5 = -\Omega^2 A^3$, we obtain 
    \begin{equation*}
        A^{2(n+1)+3} = A^2 A^{2n+3} = A^2 (-1)^n \Omega^{2n} A^3 = (-1)^n \Omega^{2n} A^5 = (-1)^{n+1} \Omega^{2(n+1)}A^3
    \end{equation*}
    where we used \eqref{eq: recursion} at $n$. Therefore, \eqref{eq: recursion} also holds at $n+1$, which concludes the recursion.    
\end{proof}

Using Lemma~\ref{lemma: recursion}, we can now calculate the exponential series:
\begin{align*}
    e^{At} - I - At &= \sum_{n\, =\, 2}^{+\infty} A^n \frac{t^n}{n!} = \sum_{n\, =\, 0}^{+\infty} A^{2n+2} \frac{t^{2n+2}}{(2n+2)!} + \sum_{n\, =\, 0}^{+\infty} A^{2n+3} \frac{t^{2n+3}}{(2n+3)!} \\
    &= \sum_{n\, =\, 0}^{+\infty} (-1)^n \Omega^{2n} A^2 \frac{t^{2n+2}}{(2n+2)!} + \sum_{n\, =\, 0}^{+\infty} (-1)^n \Omega^{2n} A^{3}\frac{t^{2n+3}}{(2n+3)!} \\
    &= -\frac{A^2}{\Omega^2} \sum_{n\, =\, 0}^{+\infty} (-1)^{n+1} \frac{ (\Omega t)^{2n+2} }{(2n+2)!} - \frac{A^{3}}{\Omega^3} \sum_{n\, =\, 0}^{+\infty} (-1)^{n+1} \frac{(\Omega t)^{2n+3}}{(2n+3)!} \\
    &= -\frac{A^2}{\Omega^2} \big( \cos(\Omega t) - 1 \big) - \frac{A^{3}}{\Omega^3} \big( \sin(\Omega t) - \Omega t \big) = \frac{A^2}{\Omega^2} \big( 1- \cos(\Omega t) \big) + \frac{A^{3}}{\Omega^3} \big( \Omega t - \sin(\Omega t) \big).
\end{align*}
Then, replacing $A$, $A^2$ and $A^3$ in the equation above, we obtain:
\begin{equation*}
    e^{At} = \begin{bmatrix} 4 -3 \cos(\Omega t) & 0 & \frac{1}{\Omega}\sin(\Omega t) & \frac{2}{\Omega}\big(1 - \cos(\Omega t) \big) \\ 6\big(\sin(\Omega t)-\Omega t\big) & 1 & \frac{2}{\Omega}\big(\cos(\Omega t) - 1\big) & -3t + \frac{4}{\Omega}\sin(\Omega t) \\ 3\Omega\sin(\Omega t) & 0 & \cos(\Omega t) & 2\sin(\Omega t) \\ 6\Omega \big(\cos(\Omega t) - 1\big) & 0 & -2\sin(\Omega t) & 4\cos(\Omega t) -3 \end{bmatrix}. 
\end{equation*}
After the loss of control authority over thruster no. 4, matrix $C$ was defined in \eqref{eq:B and C} and matrix $R_\theta$ was defined in \eqref{eq:ODE rotation}. For the reader's convenience we restate both of these matrices here as
\begin{equation*}
    C = \left[\def\arraystretch{0.6}\begin{array}{c} 0 \\ 0 \\ -\sqrt{2} \\ 0 \end{array}\right] \hspace{-1mm}, \quad R_\theta(t) = \begin{bmatrix} 1 & 0 & 0 & 0 \\ 0 & 1 & 0 & 0 \\ 0 & 0 & \cos\big( \theta(t) \big) & -\sin\big( \theta(t) \big) \\ 0 & 0 & \sin\big( \theta(t) \big) & \cos\big( \theta(t) \big) \end{bmatrix} \quad \text{and} \quad R_\theta(t) C = -\sqrt{2} \begin{bmatrix} 0 \\ 0 \\ \cos\big(\theta(t) \big) \\ \sin\big(\theta(t)\big) \end{bmatrix} \hspace{-1mm}.
\end{equation*}
Then,
\begin{equation*}
  e^{A T} R_\theta(t) C = -\frac{\sqrt{2}}{\Omega} \begin{bmatrix} \sin(\Omega T)\cos\big(\theta(t)\big) + 2\big( 1 -\cos(\Omega T) \big)\sin\big(\theta(t)\big) \\ 2\big( \cos(\Omega T) - 1\big)\cos\big(\theta(t)\big) + \big(4\sin(\Omega T) -3\Omega T\big)\sin\big(\theta(t)\big) \\ \Omega\cos(\Omega T) \cos\big(\theta(t) \big) + 2\Omega\sin(\Omega T)\sin\big(\theta(t)\big) \\ -2\Omega \sin(\Omega T) \cos\big(\theta(t) \big) + \Omega\big(4\cos(\Omega T) -3\big)\sin\big(\theta(t)\big) \end{bmatrix}.
\end{equation*}
Since $R_\theta$ is a rotation matrix, its inverse can be easily calculated as
\begin{equation*}
    R_\theta^{-1}(t+T) = \begin{bmatrix} 1 & 0 & 0 & 0 \\ 0 & 1 & 0 & 0 \\ 0 & 0 & \cos\big( \theta(t+T) \big) & \sin\big( \theta(t+T) \big) \\ 0 & 0 & -\sin\big( \theta(t+T) \big) & \cos\big( \theta(t+T) \big) \end{bmatrix}.
\end{equation*}
Since the first two rows of $B$ are null, to have $-R_\theta^{-1}(t+T)e^{AT}R_\theta(t)C\mathcal{W} \subseteq B\mathcal{U}$, the first two components of $-R_\theta^{-1}(t+T)e^{AT}R_\theta(t)C$ should also be zero for all $t \geq 0$, i.e.,
\begin{align}
    0 &= \sin(\Omega T) \cos\big(\theta(t)\big) + 2\big(1 - \cos(\Omega T) \big)\sin\big(\theta(t)\big) , \label{eq: c1} \\ 
    0 &= 2 \big( \cos(\Omega T)-1\big)\cos\big(\theta(t)\big) + \big( 4\sin(\Omega T) -3\Omega T \big) \sin\big(\theta(t)\big). \label{eq: c2}
\end{align}
For \eqref{eq: c1} to hold independently of $\theta(t)$, we need $\sin(\Omega T) = 0$ and $\cos(\Omega T) = 1$, i.e., $T = \frac{2 \pi}{\Omega}n$, $n \in \mathbb{N}$. However, \eqref{eq: c2} would yield $\sin\big(\theta(t)\big) = 0$, i.e., $y(t) = 0$ for all $t \geq 0$ which prevents the spacecraft from tracking trajectory $\mathcal{T}_\text{ref}$. Therefore, we cannot define a minimal correction time $T_c$ for the nonlinear spacecraft dynamics \eqref{eq:delayed split ODE rotation}. Then, we cannot cancel exactly $Cw(t)$ after some actuation delay as we did in Section~\ref{subsec: delayed res traj tracking}. Without this perfect cancellation an open-loop controller like in Theorem~\ref{thm: delayed res traj track} would not be able to track a trajectory. We will then transform this controller into a more complex feedback controller.

\subsection{Closed-loop controller}\label{subsec: feedback}

Motivated by the ISS malfunction \citep{ISS_thruster} where the undesirable thrust was constant, we will assume in this section that $w$ is Lipschitz continuous. As in Theorem~\ref{thm: delayed res traj track} we partition $\mathcal{P} = B\mathcal{U} \ominus -C\mathcal{W}$ into two parts: $\mathcal{P}_\text{ref}$ for the trajectory tracking and $\mathcal{P}_\varepsilon$ for the feedback correction. However, $u(t)$ has only access to $X(t-\tau)$ and not $X(t)$, hence a straightforward linear feedback is not possible. We will replace $X(t)$ by a state predictor $X_p(t)$, designed to predict $X(t)$ based on the information available at time $t-\tau$. We will use a predictor adapted from \citep{Lechappe} which takes advantage of the system's dynamics:
\begin{equation}\label{eq:predictor}
    X_{p}(t) = e^{A\tau} X(t-\tau) + \int_{t-\tau}^{t} e^{A(t-s)} r R_\theta(s) \big( B u(s) + Cw(s-\tau) \big) ds.
\end{equation}

Before stating our main theorem for resilient trajectory tracking, we recall the definitions of $\rho_{max}$ from Proposition~\ref{prop:resilience}, $\rho_{max} = \max\big\{ \rho \geq 0 : \mathbb{B}^2(0, \rho) \subseteq \mathcal{P} \big\}$, and the log-norm $\mu(A) = \max\big\{ \Lambda((A+A^\top)/2)\big\}$.

\begin{theorem}\label{thm: closed loop res track}
    Let $K \in \mathbb{R}^{4 \times 4}$ such that $\tilde{A} := A - rBK$ is Hurwitz, and let $P \in \mathbb{R}^{4 \times 4}$ and $Q \in \mathbb{R}^{4 \times 4}$ such that $P \succ 0$, $Q \succ 0$ and $\tilde{A}^\top P + P \tilde{A} = -Q$. Define $\alpha := \frac{\lambda_{min}^{Q}}{2\lambda_{max}^{P}}$, $\beta := r\sqrt{\lambda_{max}^P}\|C\|L\tau$, and $\gamma := r\|BK\|\frac{e^{\mu(A)\tau} -1}{\mu(A)}$. For $L > 0$, let $\varepsilon := \frac{\|BK\|}{\sqrt{\lambda_{min}^P}} \max\Big( \|X_\text{ref}(0) - X(0)\|_P, \frac{\beta}{\alpha}(1+\gamma)\Big) + \gamma \|C\|L\tau$.
    
    If $\varepsilon + \rho_\text{ref} \leq \rho_{max}$, then, for all $w \in \mathcal{F}(\mathcal{W})$ with a Lipschitz constant $L$, the malfunctioning spacecraft \eqref{eq:delayed split ODE rotation} can track reference trajectory $\mathcal{T}_\text{ref}$ with a tolerance $\big\| X_\text{ref}(t) - X(t)\big\| \leq \frac{1}{\sqrt{\lambda_{min}^P}} \max \left( \|X_\text{ref}(0) - X(0)\|_P,\ \frac{\beta}{\alpha}(1+\gamma) \right)$ for all $t \geq \tau$.
\end{theorem}
\begin{proof}
    The existence of matrix $K$ is justified by the controllability of the pair $(A,B)$ \citep{Khalil}. Since the resulting $\tilde{A}$ is Hurwitz, matrices $P \succ 0$ and $Q \succ 0$ exist according to Lyapunov theory \citep{Kalman}. We consider any $w \in \mathcal{F}(\mathcal{W})$ with a Lipschitz constant $L$ and assume that $\varepsilon +\rho_\text{ref} \leq \rho_{max}$. We define $\mathcal{P}_{\varepsilon} := \mathbb{B}^2(0, \varepsilon)$ and recall $\mathcal{P}_b =\mathbb{B}^2(0, \rho_{max})$ as introduced in Proposition~\ref{prop:resilience}. Then, $\mathcal{P}_\varepsilon \oplus \mathcal{P}_\text{ref} = \mathbb{B}^2(0, \varepsilon) \oplus \mathbb{B}^2(0, \rho_\text{ref}) = \mathbb{B}^2(0, \varepsilon+\rho_\text{ref}) \subseteq \mathbb{B}^2(0, \rho_{max}) = \mathcal{P}_b \subseteq \mathcal{P}$.
    
    For $t \geq \tau$ we introduce control signals $u_w$, $u_\text{ref}$ and $u_\varepsilon$ such that $Bu_w(t) := -Cw(t-\tau)$, $Bu_\text{ref}(t) := R_\theta^{-1}(t) p_\text{ref}(t)$ and $Bu_\varepsilon(t) := R_\theta^{-1}(t) BK\big(X_\text{ref}(t) - X_p(t)\big)$ with the predictor $X_p$ from \eqref{eq:predictor}. We consequently define the feedback control law $u$ by
    \begin{equation}\label{eq: feedback controller}
        Bu(t) := Bu_w(t) + Bu_\text{ref}(t) + Bu_\varepsilon(t) = -Cw(t-\tau) + R_\theta^{-1}(t) p_\text{ref}(t) + R_\theta^{-1}(t) BK\big(X_\text{ref}(t) - X_p(t)\big).
    \end{equation}
    To prove that controller \eqref{eq: feedback controller} is admissible we need to show that $Bu(t) \in B\mathcal{U}$ for all $t \geq \tau$. Firstly, $Bu_w(t) = -Cw(t - \tau) \in -C \mathcal{W}$.
    Because $\mathcal{P}_\text{ref}$ is a ball centered on $0$, it is invariant by rotation $R_\theta$. Then, $R_\theta^{-1}(t) p_\text{ref}(t) \in \mathcal{P}_\text{ref}$, i.e., $Bu_\text{ref}(t) \in \mathcal{P}_\text{ref}$. Since $-C\mathcal{W} \oplus \mathcal{P}_\text{ref} \oplus \mathcal{P}_\varepsilon \subseteq B\mathcal{U}$, it now suffices to show that $Bu_\varepsilon(t) \in \mathcal{P}_\varepsilon = \mathbb{B}^2(0, \varepsilon)$. To do so, we first apply \eqref{eq: feedback controller} to dynamics \eqref{eq:delayed split ODE rotation}. By definition of $\mathcal{T}_\text{ref}$, we have $r p_\text{ref}(t) = \dot X_\text{ref}(t) - AX_\text{ref}(t)$, and thus
    \begin{align*}
        \dot{X}(t) &= AX(t) + rR_\theta(t) Bu(t) + rR_\theta(t) Cw(t) \\
        &= AX(t) -rR_\theta(t) Cw(t-\tau) +\dot X_\text{ref}(t) - AX_\text{ref}(t) + rBK\big(X_\text{ref}(t) - X_p(t)\big) + rR_\theta(t)Cw(t),
    \end{align*}
    i.e., 
    \begin{equation*}
        \dot{X}(t) - \dot X_\text{ref}(t) = (A- rBK) \big(X(t) - X_\text{ref}(t)\big) + rR_\theta(t)\big( Cw(t) - Cw(t-\tau) \big) + rBK \big(X(t) - X_p(t)\big).
    \end{equation*}
    We define
    \begin{equation*}
        Y(t) := X(t) - X_\text{ref}(t), \quad \Delta C(t) := rR_\theta(t)\big( Cw(t) - Cw(t-\tau) \big), \quad \Delta X(t) := r BK\big( X(t) - X_p(t) \big),
    \end{equation*}
    so that $\dot Y(t) = \tilde{A} Y(t) + \Delta C(t) + \Delta X(t)$.
    Inspired by the method described in Section~9.3 of \citep{Khalil}, we will now show that $Y(t)$ is bounded, which in turn will prove that $Bu_\varepsilon(t) \in \mathcal{P}_\varepsilon$ and hence that control law \eqref{eq: feedback controller} is admissible. We consider the derivative of the norm $Y(t)^\top P Y(t) = \|Y(t)\|_{P}^2$ and obtain the following:
    \begin{equation*}
        \frac{d}{dt} \|Y(t)\|_P^2 = \dot Y(t)^\top P Y(t) + Y(t)^\top P \dot Y(t) = Y(t)^\top \big(\tilde{A}^\top P + P \tilde{A} \big) Y(t) + 2 Y(t)^\top P \big(\Delta C(t) + \Delta X(t) \big).
    \end{equation*}
    Since $\|\cdot\|_P$ is a norm, the Cauchy-Schwarz inequality \citep{Matrix} yields $Y(t)^\top P \Delta C(t) \leq \|Y(t)\|_{P} \|\Delta C(t)\|_{P}$. Then, using $\|R_\theta(t)\| = 1$ and the Lipschitz constant $L$ of $w$, we have
    \begin{align*}
        \|\Delta C(t)\|_{P} &\leq \sqrt{\lambda_{max}^P} r\|R_\theta(t)\| \|Cw(t) - Cw(t-\tau)\| \leq r\sqrt{\lambda_{max}^P} \|C\| \big| w(t) - w(t-\tau) \big| \\
        &\leq r\sqrt{\lambda_{max}^P} \|C\| L \tau = \beta.
    \end{align*}
    Similarly, 
    \begin{equation*}
        \|\Delta X(t)\|_P \leq r \sqrt{\lambda_{max}^P} \big\| BK \big( X(t) - X_p(t) \big) \big\| \leq r \sqrt{\lambda_{max}^P} \|BK\| \|X(t) - X_p(t)\|.
    \end{equation*}
    We write the state of the system $X(t)$ in a form similar to \eqref{eq:predictor} to compare it with $X_p$:
    \begin{equation*}
        X(t) = e^{A\tau} X(t-\tau) + \int_{t-\tau}^t e^{A(t-s)}rR_\theta(s) \big( Bu(s) + Cw(s) \big)ds.
    \end{equation*}
    Then, reusing the log-norm $\mu(A)$ \citep{exp} as in Theorem~\ref{thm: delayed res reach}, we obtain
    \begin{align*}
        \big\| X(t) - X_p(t) \big\| &\leq \int_{t-\tau}^{t} \big\|e^{A(t-s)}\big\| r \|R_\theta(s)\| \big\| Cw(s) - Cw(s-\tau)\big\| ds \\
        &\leq r \int_{t-\tau}^{t} e^{\mu(A)(t-s)} \|C\| L\tau\, ds = r \|C\| L\tau \frac{e^{\mu(A)\tau} -1}{\mu(A)}.
    \end{align*}
    Therefore, $\|\Delta X(t)\|_P \leq r \sqrt{\lambda_{max}^P} \|BK\| r \|C\| L \tau \frac{e^{\mu(A)\tau} -1}{\mu(A)} = \beta \gamma$, so that
    \begin{equation*}
        \frac{d}{dt} \|Y(t)\|_P^2 \leq -Y^\top(t) Q Y(t) + 2 \|Y(t)\|_{P} \big( \|\Delta C(t)\|_P +\|\Delta X(t)\|_P\big) \leq -\frac{\lambda_{min}^{Q}}{\lambda_{max}^{P}} \|Y(t)\|_{P}^2 + 2 \beta(1+\gamma) \|Y(t)\|_{P}.
    \end{equation*}
    Indeed, $Q \succ 0$ yields $-Y^\top Q Y \leq -\lambda_{min}^{Q} Y^\top Y$ \citep{Khalil} and $\|Y\|_{P}^2 \leq \lambda_{max}^{P} Y^\top Y$ leads to $-Y^\top Y \leq \frac{-1}{\lambda_{max}^{P}}\|Y\|_{P}^2$. Hence, we obtain
    \begin{equation*}
        \frac{d}{dt} \|Y(t)\|_{P}^2 \leq -2\alpha \|Y(t)\|_{P}^2 + 2\beta(1+\gamma) \|Y(t)\|_{P}.
    \end{equation*}
    Since $\frac{d}{dt} \|Y(t)\|_{P}^2 = 2\|Y(t)\|_{P} \frac{d}{dt} \|Y(t)\|_{P}$, we have $\frac{d}{dt} \|Y(t)\|_{P} \leq -\alpha \|Y(t)\|_{P} + \beta(1+\gamma)$ for $Y(t) \neq 0$. Let us define the function $f(v) := -\alpha v + \beta(1+\gamma)$. The solution of the differential equation $\dot v(t) = f\big( v(t) \big)$ with initial condition $v(0) = \|Y(0)\|_{P}$ is $v(t) = e^{-\alpha t} \left( \|Y(0)\|_{P} - \frac{\beta}{\alpha}(1+\gamma) \right) + \frac{\beta}{\alpha}(1+\gamma)$. Since $f(v)$ is Lipschitz in $v$ and $\frac{d}{dt} \|Y(t)\|_{P} \leq f\big( \|Y(t)\|_{P} \big)$, the Comparison Lemma of \citep{Khalil} states that $\|Y(t)\|_{P} \leq v(t)$ for all $t \geq 0$. Then,
    \begin{equation*}
        \|Y(t)\|_P \leq e^{-\alpha t} \left( \|Y(0)\|_{P} - \frac{\beta}{\alpha}(1+\gamma) \right) + \frac{\beta}{\alpha}(1+\gamma) \xrightarrow[t \rightarrow \infty]{} \frac{\beta}{\alpha}(1+\gamma). 
    \end{equation*}
    Since this bound on $\|Y(t)\|_P$ is monotonic, we have $\|Y(t)\|_P \leq \max \left( \|Y(0)\|_P,\ \frac{\beta}{\alpha}(1+\gamma) \right)$. Then,
    \begin{align*}
        \|Bu_\varepsilon(t)\| &= \big\| R_\theta^{-1}(t) BK\big(X_\text{ref}(t) - X_p(t)\big) \big\| \leq \big\|R_\theta^{-1}(t)\big\| \, \|BK\| \, \big(\|X_\text{ref}(t) - X(t)\| + \|X(t) - X_p(t)\|\big) \\
        &\leq \|BK\| \left( \|Y(t)\| +  r \|C\| L\tau \frac{e^{\mu(A)\tau} -1}{\mu(A)} \right) \leq \|BK\| \frac{\|Y(t)\|_P}{\sqrt{\lambda_{min}^P}} + \gamma \|C\|L\tau \leq \varepsilon,
    \end{align*}
    by definition of $\varepsilon$ and using $\big\|R_\theta^{-1}(t)\big\| = 1$. Therefore, $Bu_\varepsilon(t) \in \mathbb{B}^2(0, \varepsilon) = \mathcal{P}_\varepsilon$.
    To sum up, $Bu(t) = Bu_w(t) + Bu_\text{ref}(t) + Bu_\varepsilon(t) \in -C\mathcal{W} \oplus \mathcal{P}_\text{ref} \oplus \mathcal{P}_\varepsilon \subseteq -C\mathcal{W} \oplus \mathcal{P} \subseteq B\mathcal{U}$ for all $t \geq 0$. Therefore, control law \eqref{eq: feedback controller} is admissible and the announced tracking tolerance is verified: 
    \begin{equation*}
        \|X_\text{ref}(t) - X(t)\| = \|Y(t)\| \leq \frac{\|Y(t)\|_P}{\sqrt{\lambda_{min}^P}} \leq  \max \left( \frac{\|Y(0)\|_P}{\sqrt{\lambda_{min}^P}},\ \frac{\beta+\beta\gamma}{\alpha\sqrt{\lambda_{min}^P}} \right).
    \end{equation*}
\end{proof}

Theorem~\ref{thm: closed loop res track} provides a controller with trajectory tracking guarantees for the malfunctioning spacecraft \eqref{eq:delayed split ODE rotation}. The tracking error is dictated by two main terms $\beta$ and $\beta \gamma$, which respectively bound the prediction errors on the undesirable thrust $\|\Delta C(t)\|_P$ and the state $\|\Delta X(t)\|_P$. The term $\alpha\sqrt{\lambda_{min}^P}$ is just a conversion factor between the $P$-norm and the Euclidean norm. 
The term $\|X_\text{ref}(0) - X(0)\|_P$ in the definition of $\varepsilon$ ensures that the controller is robust to initial state uncertainty as discussed in Section~\ref{subsec: tracking no delay}.
The tracking tolerance of Theorem~\ref{thm: closed loop res track} can also be interpreted as a convergence radius for the controller. Indeed, in the proof of Theorem~\ref{thm: closed loop res track} we showed that $\|X_\text{ref}(t) - X(t)\|$ must be small enough for $u_\varepsilon(t)$ to be admissible. If $\|X_\text{ref}(0) - X(0)\|$ is too large, control law \eqref{eq: feedback controller} might not be admissible and the convergence of $X(t)$ to $\mathcal{T}_\text{ref}$ cannot be guaranteed. The choice of matrices $K$, $P$ and $Q$ must then be optimized for the controller to be sufficiently robust to initial state uncertainty. A similar but simplified optimization is discussed in Exercise~9.1 of \citep{Khalil}.

We will now implement controller \eqref{eq: feedback controller} embedded with predictor \eqref{eq:predictor} on the malfunctioning spacecraft dynamics \eqref{eq:delayed split ODE rotation} to study its performance over the course of the inspection mission and respond to Problem~\ref{prob:mission}.

\section{Numerical simulation}\label{sec:simulation}

In this section we study whether controller \eqref{eq: feedback controller} can fulfill the mission scenario of Section~\ref{sec:motivation}. 
Recall that the statement of Problem~\ref{prob:mission} specifies neither the malfunctioning thruster, nor the regularity of the undesirable thrust signal $w$, nor the value of the actuation delay $\tau$. 
As discussed above Fig.~\ref{fig: p_ref}, tracking the reference trajectory of Fig.~\ref{fig:ref orbit} appears to only be possible if the malfunctioning thruster is no. 4. Therefore, we will only investigate scenarios featuring the loss of control authority over thruster no. 4. In such a case $\rho_{max} > 0$, which enables us to apply Theorem~\ref{thm: closed loop res track} and use controller~\eqref{eq: feedback controller}. Then, to address Problem~\ref{prob:mission} we will simulate a variety of scenarios with different undesirable thrust signals and different actuation delays. We perform all the simulations in MATLAB and all the codes are accessible on github\footnote{\url{https://github.com/Jean-BaptisteBouvier/Spacecraft-Resilience}}.

\subsection{Nominal scenario}\label{subsec: w rand 0.2s}

In this first scenario, we choose an actuation delay $\tau = 0.2\, s$ following \citep{thruster_delay} and a Lipschitz constant $L = 0.1$ for $w$ so that the malfunctioning thrust cannot vary by more than a tenth of its capability every second since $\mathcal{W} = [0, 1]$. We choose matrices $K$, $P$ and $Q$ to maximize $\varepsilon$ subject to $\varepsilon \leq \rho_{max} - \rho_\text{ref}$, where $\rho_\text{ref} = 4.85 \times 10^{-4}$ is the maximal input norm on the reference trajectory, as seen on Fig.~\ref{fig: p_ref}. Ample numerical testing on MATLAB led us to believe that the optimal matrices are $Q = I$ and $K$ such that $A - r B K$ has 4 identical eigenvalues. Then,
\begin{equation*}
    P = \begin{bmatrix} 2.77 & 0 & 1.77 & 0.01 \\ 0 & 2.77 & -0.01 & 1.77 \\ 1.77 & -0.01 & 8 & 0 \\ 0.01 & 1.77 & 0 & 8 \end{bmatrix} \quad \text{and} \quad K = 472 \begin{bmatrix} 1 & 1 & 1 & 1 \\ 1 & -1 & 1 & -1\\ -1 & -1 & -1 & -1 \\ -1 & 1 & -1 & 1 \end{bmatrix}
\end{equation*}
so that $\varepsilon = 0.4133 < \rho_{max} - \rho_\text{ref} = 0.4137$ and the tracking tolerance is $\frac{\beta(1+\gamma)}{\alpha\sqrt{\lambda_{min}^P}} = 1.5 \times 10^{-4}$ for $X(0) = X_\text{ref}(0)$. 

Then, controller~\eqref{eq: feedback controller} ensures excellent tracking of the reference trajectory $\mathcal{T}_\text{ref}$, as shown on Fig.~\ref{fig:traj and error}(\subref{fig:tracking comparison}). We compute the position error between the reference state and the tracking state on Fig.~\ref{fig:traj and error}(\subref{fig:pos comp w rand}). We observe that the position error is never larger than $1.07\, mm$ and averages only $0.36\, mm$. We acknowledge that these extremely small errors are only possible because all dynamics, states and thrusts are known exactly in our simple simulation.

\begin{figure}[htbp!]
    \centering
    \begin{subfigure}[]{0.55\textwidth}
       \includegraphics[scale = 0.48]{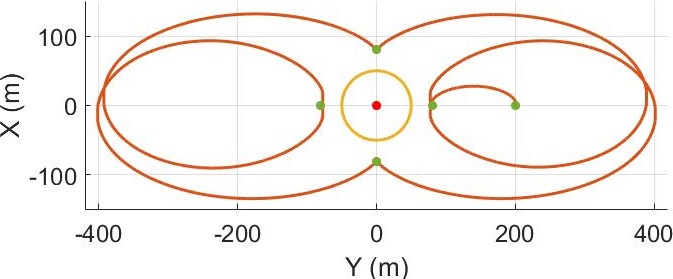}
        \caption{Trajectory tracking by controller \eqref{eq: feedback controller} (red) linking the waypoints (green) to inspect the target satellite (red) without breaching the KOS (yellow).}
        \label{fig:tracking comparison}
    \end{subfigure}\hfill
    \begin{subfigure}[]{0.44\textwidth}
        \includegraphics[scale = 0.4]{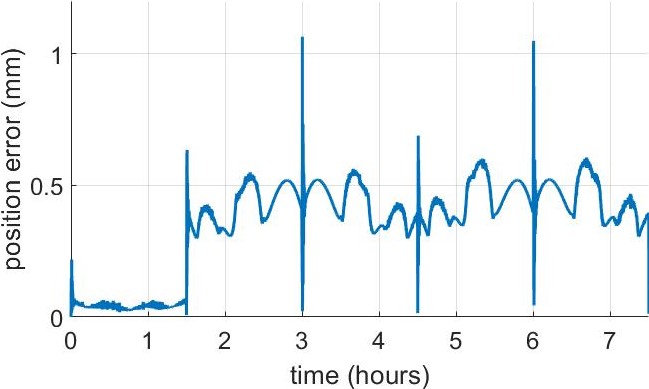}
        \caption{Position error between the reference and tracking trajectories.}
        \label{fig:pos comp w rand}
    \end{subfigure}
    \caption{Analysis of the trajectory tracking performance for a Lipschitz continuous undesirable input $w$ and actuation delay $\tau = 0.2\, s$.}
    \label{fig:traj and error}
\end{figure}

To compare with the tracking tolerance of $1.5 \times 10^{-4}$, we also compute the norm difference between the reference and tracking states: $\|X(t) - X_\text{ref}(t)\|$. The average norm difference is $3.6\times 10^{-4}$, while the maximal norm difference is $10.5 \times 10^{-4}$. Let us investigate why these values are slightly larger than the tracking tolerance. 
First, note that $\|X(t)\|^2 = x(t)^2 + y(t)^2 + \dot x(t)^2 + \dot y(t)^2$ where position $(x(t), y(t))$ is of the order of $10^2\, m$ as shown on Fig.~\ref{fig:traj and error}(\subref{fig:tracking comparison}), while velocity $(\dot x(t), \dot y(t))$ is of the order $10^{-1}\, m \cdot s^{-1}$. Because of these orders of magnitudes, the norm difference between reference and tracking states reflects mostly the position error. Based on Fig.~\ref{fig:traj and error}(\subref{fig:pos comp w rand}), the maximal norm difference occurs at $3$ hours and $6$ hours, i.e., at the waypoints $x = \pm 80\, m$ and $y = 0\, m$ as shown on Fig.~\ref{fig:traj and error}(\subref{fig:tracking comparison}). At every other waypoint Fig.~\ref{fig:traj and error}(\subref{fig:pos comp w rand}) also shows error spikes albeit of smaller magnitude. Because the sudden stop and start occurring at each waypoint are not well captured by the discrete dynamics of our simulation, the actual norm difference is larger than the threshold value of Theorem~\ref{thm: closed loop res track}.

The undesirable thrust input $w$ is generated as a stochastic signal, whose magnitude is represented in yellow in Fig.~\ref{fig:inputs}(\subref{fig:total_inputs}). To counteract $w$ while following the reference trajectory, the controlled input $u$ verifies approximately the intuitive relation $\|u\| \approx \|u_\text{ref}\| + \|w\|$. More specifically, Fig.~\ref{fig:inputs}(\subref{fig:thrusters}) shows that thrust inputs $u_3$ and $u_5$ replicate the reference thrust profile of Fig.~\ref{fig:ref_input}, while $u_1$ and $u_2$ counteract malfunctioning thruster no. 4 as expected from their opposite placement on Fig.~\ref{fig:spacecraft}.

\begin{figure}[htbp!]
    \centering
    \begin{subfigure}[]{0.49\textwidth}
        \includegraphics[scale = 0.45]{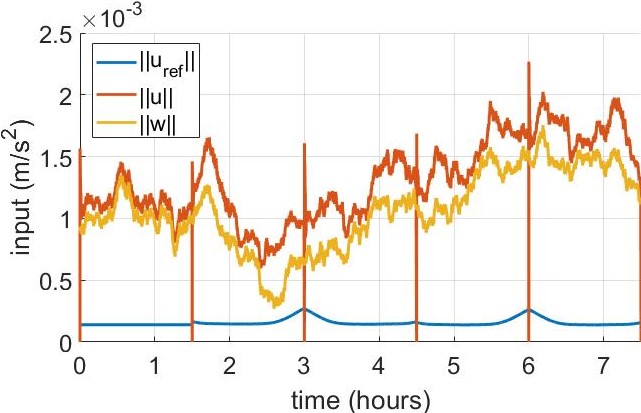}
        \caption{Magnitude of the thrust inputs for the reference trajectory $\|u_\text{ref}\|$ (blue), for the tracking trajectory the controlled input is $\|u\|$ (red) and the undesirable input is $\|w\|$ (yellow).}
        \label{fig:total_inputs}
    \end{subfigure}\hfill
    \begin{subfigure}[]{0.49\textwidth}
        \includegraphics[scale = 0.45]{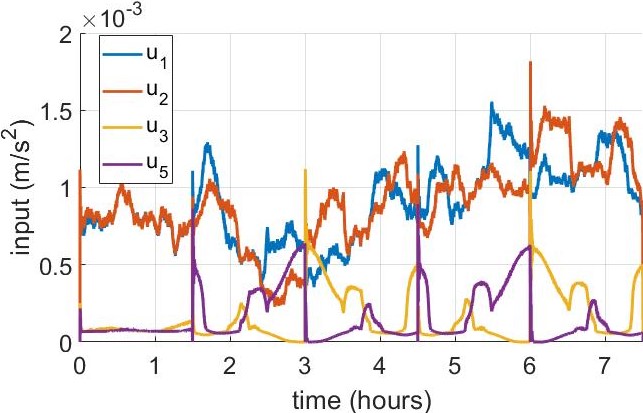}
        \caption{Thrust profiles for the four controlled thrusters of the chaser satellite on the tracking trajectory.}
        \label{fig:thrusters}
    \end{subfigure}
    \caption{Analysis of the thrust profiles of the malfunctioning satellite for a Lipschitz undesirable input $w$ and actuation delay $\tau = 0.2\, s$.}
    \label{fig:inputs}
\end{figure}

The fuel consumption on the reference and tracking trajectories is displayed on Fig.~\ref{fig:fuel and velocity}(\subref{fig:fuel}). The yellow curve represents the mass of fuel $m_w = 1.06\, kg$ used to produce the undesirable thrust, while the red one shows the mass of fuel $m_u = 1.31\, kg$ used by the controlled thrusters. The reference trajectory without malfunctions requires $m_\text{ref} = 0.16\, kg$ of fuel. As expected, $m_u \approx m_\text{ref} + m_w$. We have the intuition that the gap between $m_u$ and $m_\text{ref} + m_w$ will grow with $\tau$ and with the unpredictability of $w$. 

As can be expected from the tracking accuracy displayed on Fig.~\ref{fig:traj and error}, the velocity tracking of the reference is also extremely accurate with velocities remaining within $0.35\, mm/s$ of each others, as illustrated on Fig.~\ref{fig:fuel and velocity}(\subref{fig:velocity}). As on Fig.~\ref{fig:traj and error}(\subref{fig:pos comp w rand}), the error spikes at each waypoint and displays also the same symmetry as the orbit.

\begin{figure}[htbp!]
    \centering
    \begin{subfigure}[]{0.49\textwidth}
        \includegraphics[scale = 0.45]{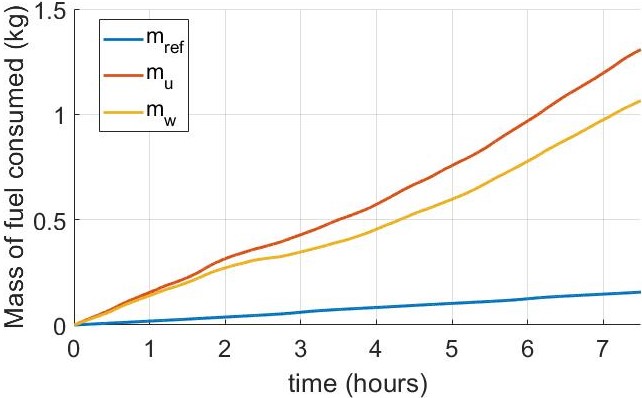}
        \caption{Comparison of fuel consumption. The mass of fuel used to complete the reference trajectory is $m_\text{ref}$ (blue). After the loss of control over thruster no. 4, it consumes a mass of fuel $m_w$ (yellow), while the controlled thrusters use a mass $m_u$ (red).}
        \label{fig:fuel}
    \end{subfigure}\hfill
    \begin{subfigure}[]{0.49\textwidth}
        \includegraphics[scale = 0.45]{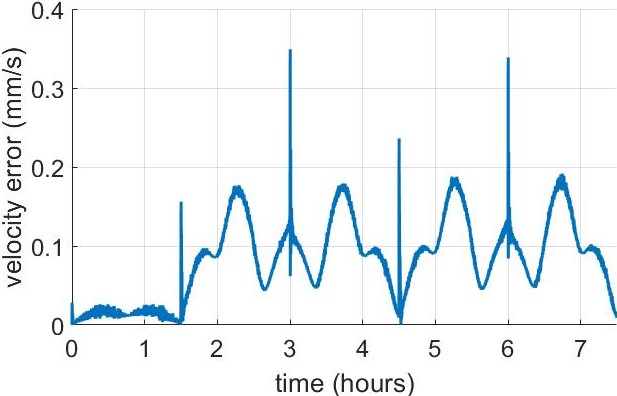}
        \caption{Velocity error between the reference and tracking trajectories.}
        \label{fig:velocity}
    \end{subfigure}
    \caption{Comparison of the fuel consumption and velocities for a Lipschitz continuous undesirable input $w$ and actuation delay $\tau = 0.2\, s$.}
    \label{fig:fuel and velocity}
\end{figure}

Based on Fig.~\ref{fig:traj and error}, \ref{fig:inputs} and \ref{fig:fuel and velocity}, controller~\eqref{eq: feedback controller} performs excellently and enables the system to complete its mission when $w$ is Lipschitz continuous and $\tau = 0.2\, s$. To address Problem~\ref{prob:mission} let us now investigate more challenging scenarios.

\subsection{Lipschitz undesirable thrust and actuation delay of 8 seconds}\label{subsec: w rand 8s}

We first increase the actuation delay $\tau$ from $0.2\, s$ to $8\, s$ and keep the same Lipschitz and stochastic undesirable thrust signal $w$. The previous guarantees of Theorem~\ref{thm: closed loop res track} are not valid anymore, but controller~\eqref{eq: feedback controller} still performs sufficiently well to not be distinguishable from the reference as in Fig.~\ref{fig:traj and error}(\subref{fig:tracking comparison}). Instead, we analyze the position error shown on Fig.~\ref{fig:w rand}(\subref{fig:pos comp w rand 8s}). The trajectory tracks the reference with an average position error of $1.2\, mm$ and a maximal error of $4.1\, mm$. These values are extremely low but still represent a fourfold increase compared to the scenario with $\tau = 0.2\, s$.

\begin{figure}[htbp!]
    \centering
    \begin{subfigure}[]{0.49\textwidth}
        \includegraphics[scale = 0.45]{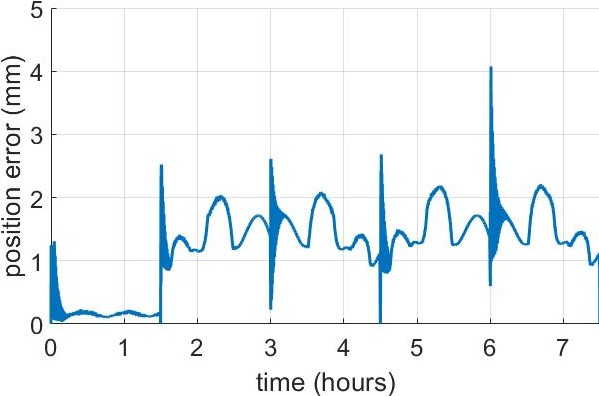}
        \caption{Position error between the reference and tracking trajectories.}
        \label{fig:pos comp w rand 8s}
    \end{subfigure}\hfill
    \begin{subfigure}[]{0.49\textwidth}
        \includegraphics[scale = 0.35]{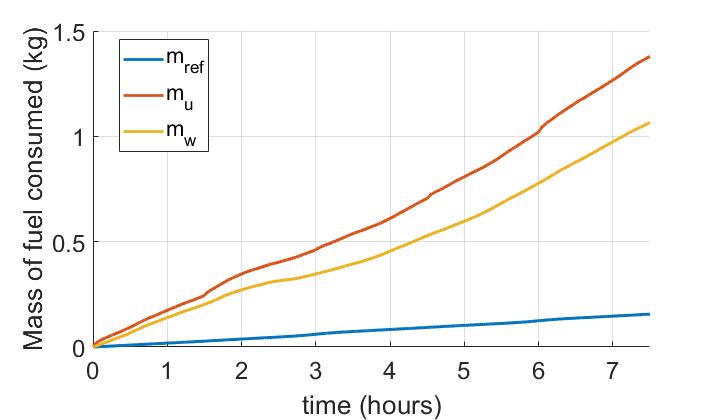}
        \caption{Comparisons of fuel used to complete the reference trajectory $m_\text{ref}$ (blue), the tracking trajectory $m_u$ (red), and the fuel used by malfunctioning thruster no. 4 $m_w$ (yellow).}
        \label{fig:fuel comp w rand 8s}
    \end{subfigure}
    \caption{Analysis of the trajectory tracking performance for a stochastic Lipschitz undesirable input $w$ and actuation delay $\tau = 8\, s$.}
    \label{fig:w rand}
\end{figure}

Concerning the fuel efficiency, the pseudo-equality $m \approx m_\text{ref} + m_w$ derived from Fig.~\ref{fig:fuel and velocity}(\subref{fig:fuel}) still holds approximately since $m_\text{ref} = 0.16\, kg$, $m_w = 1.06\, kg$ and $m_u = 1.38\, kg$ in this scenario. The controlled thrusters have only slightly increased their consumption compared to $m_u = 1.31\, kg$ for $\tau = 0.2\, s$. Thus, the actuation delay does not play as crucial a role for the fuel consumption as for the position error.

\subsection{Lipschitz undesirable thrust and actuation delay of 10 seconds}\label{subsec: w rand 10s}

If we increase further the actuation delay, e.g. $\tau = 10\, s$, controller~\eqref{eq: feedback controller} becomes incapable of tracking the reference trajectory as depicted on Fig.~\ref{fig: w rand 10s}(\subref{fig:velocity w rand 10s}). The velocity on the tracking trajectory is on average four times larger than the reference.

\begin{figure}[htbp!]
    \centering
    \begin{subfigure}[]{0.49\textwidth}
        \includegraphics[scale = 0.35]{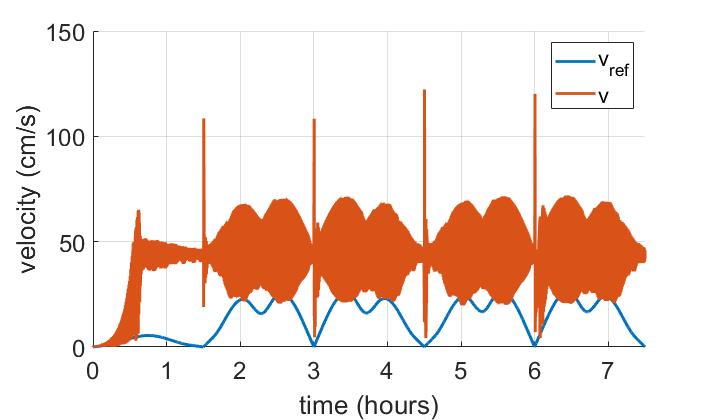}
        \caption{Velocity comparison for the reference trajectory $v_\text{ref}$ (blue) and the tracking trajectory $v$ (red).}
        \label{fig:velocity w rand 10s}
    \end{subfigure}\hfill
    \begin{subfigure}[]{0.49\textwidth}
        \includegraphics[scale = 0.45]{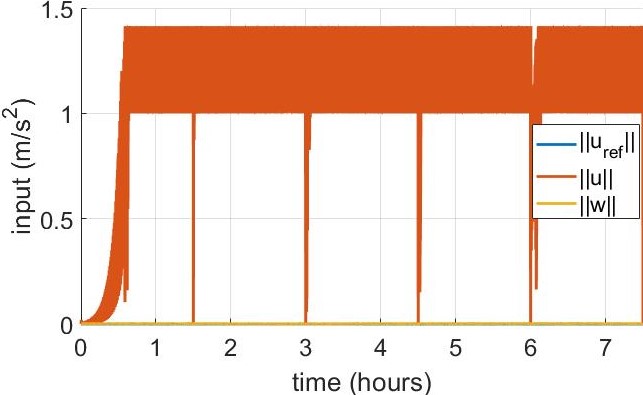}
        \caption{Comparison of input magnitudes with a saturated control $\|u\|$ (red) orders of magnitude larger than the reference $\|u_\text{ref}\|$ (blue) and the undesirable input $\|w\|$ (yellow).}
        \label{fig:inputs w rand 10s}
    \end{subfigure}
    \caption{Analysis of the trajectory tracking performance for a stochastic Lipschitz undesirable input $w$ and actuation delay $\tau = 10\, s$.}
    \label{fig: w rand 10s}
\end{figure} 

The position error has also steeply increased compared to the scenario where $\tau = 8\, s$ since here the average position error is $0.48\, m$ and the maximal error is $3\, m$. These values are still small enough to keep the tracking trajectory indistinguishable from the reference on a figure like Fig.~\ref{fig:traj and error}(\subref{fig:tracking comparison}). However, to maintain this accuracy, controller~\eqref{eq: feedback controller} had to saturate its thrust inputs as shown on Fig.~\ref{fig: w rand 10s}(\subref{fig:inputs w rand 10s}). This input saturation results in a prohibitive fuel consumption of $503\, kg$ compared to $m_u = 1.38\, kg$ for $\tau = 8\, s$. 
Now that we have probed the limits of controller~\eqref{eq: feedback controller} in terms of actuation delay, let us investigate the impact of the regularity of $w$ on the tracking performance.

\subsection{Bang-bang undesirable thrust and actuation delay of 1 second}\label{subsec: w bang 1s}

In this scenario we keep the actuation delay $\tau = 1\, s$, but the undesirable thrust signal $w$ is now bang-bang, as illustrated on Fig.~\ref{fig:w bang inputs}(\subref{fig:inputs w bang 1s}). This violates the Lipschitz assumption of Theorem~\ref{thm: closed loop res track} and hence invalidates its performance guarantees.

Controller~\eqref{eq: feedback controller} generates a trajectory with an average position error of $0.54\, mm$ and a maximal error of $5.6\, mm$ as shown on Fig.~\ref{fig:w bang}(\subref{fig:pos comp w bang 1s}). These values are comparable to the precision achieved in the scenario where $w$ was Lipschitz and $\tau = 8\, s$. As expected, increasing the unpredictability of $w$ from Lipschitz to bang-bang led to a degradation of the tracking performance.
Concerning the fuel usage in this scenario, Fig.~\ref{fig:w bang}(\subref{fig:fuel comp w bang 1s}) shows that the bang-bang thrust signal yields a significant consumption increase to $m_w = 4.86\, kg$ compared to $1.06\, kg$ in the Lipschitz scenarios. This increase is reflected on the controller's fuel usage $m_u = 5.33\, kg$, which remains close to $m_\text{ref} + m_w = 5.02\, kg$.

\begin{figure}[ht]
    \centering
    \begin{subfigure}[]{0.49\textwidth}
        \includegraphics[scale = 0.45]{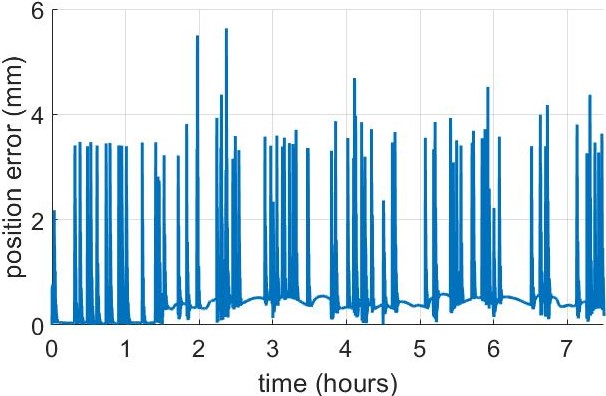}
        \caption{Position error between the reference and tracking trajectories.}
        \label{fig:pos comp w bang 1s}
    \end{subfigure}\hfill
    \begin{subfigure}[]{0.49\textwidth}
        \includegraphics[scale = 0.35]{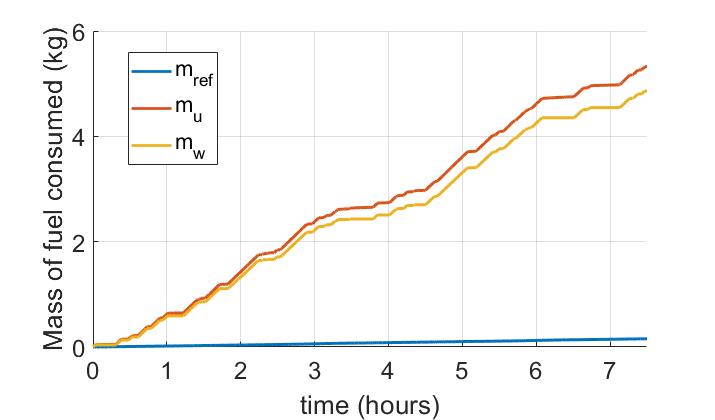}
        \caption{Comparisons of fuel used to complete the reference trajectory $m_\text{ref}$ (blue), the tracking trajectory $m_u$ (red), and the fuel used by malfunctioning thruster no. 4 $m_w$ (yellow).}
        \label{fig:fuel comp w bang 1s}
    \end{subfigure}
    \caption{Analysis of the trajectory tracking performance for a bang-bang undesirable thrust signal $w$ and actuation delay $\tau = 1\, s$.}
    \label{fig:w bang}
\end{figure}

Every time the undesirable thrust climbs to its maximum value, the controller reacts after a delay $\tau$ and with a $50\%$ higher spike to makeup for this delay, as illustrated on Fig.~\ref{fig:w bang inputs}(\subref{fig:inputs w bang 1s}). This overshoot explains the increased mass of fuel consumption by the controlled thrusters. Note also the similarity between Fig.~\ref{fig:w bang}(\subref{fig:pos comp w bang 1s}) and \ref{fig:w bang inputs}(\subref{fig:inputs w bang 1s}), each position error spike is associated with a spike of $w$.

\begin{figure}[htbp!]
    \centering
    \begin{subfigure}[]{0.49\textwidth}
        \includegraphics[scale = 0.45]{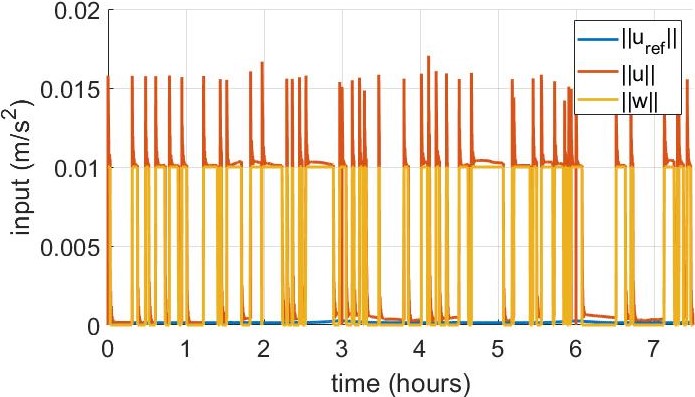}
        \caption{Actuation delay $\tau = 1\, s$.}
        \label{fig:inputs w bang 1s}
    \end{subfigure}\hfill
    \begin{subfigure}[]{0.49\textwidth}
        \includegraphics[scale = 0.45]{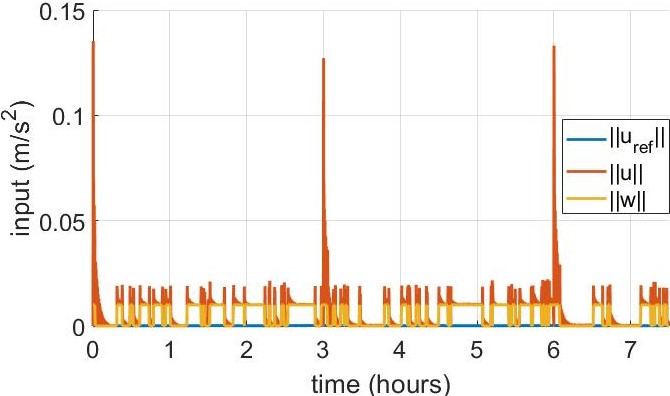}
        \caption{Actuation delay $\tau = 8\, s$.}
        \label{fig:inputs w bang 8s}
    \end{subfigure}
    \caption{Magnitude of the thrust inputs for the reference trajectory $\|u_\text{ref}\|$ (blue), the tracking trajectory $\|u\|$ (red), and the bang-bang undesirable input $\|w\|$ (yellow) for different actuation delays.}
    \label{fig:w bang inputs}
\end{figure}

Since controller~\eqref{eq: feedback controller} is still able to track the reference trajectory, we will consider a more challenging scenario with an increased actuation delay.

\subsection{Bang-bang undesirable thrust and actuation delay of 8 seconds}\label{subsec: w bang 8s}

We increase the actuation delay to $\tau = 8\, s$ while keeping the same bang-bang undesirable thrust signal as in the previous scenario. The overshoots of the controller have become much larger at three waypoints as shown on Fig.~\ref{fig:w bang inputs}(\subref{fig:inputs w bang 8s}), while the overshoots at other locations have an amplitude similar to that of $w$. These large spikes are still an order of magnitude smaller than those of Fig.~\ref{fig: w rand 10s}(\subref{fig:inputs w rand 10s}), so the controller is not saturating yet.

\begin{figure}[htbp!]
    \centering
    \begin{subfigure}[]{0.49\textwidth}
        \includegraphics[scale = 0.45]{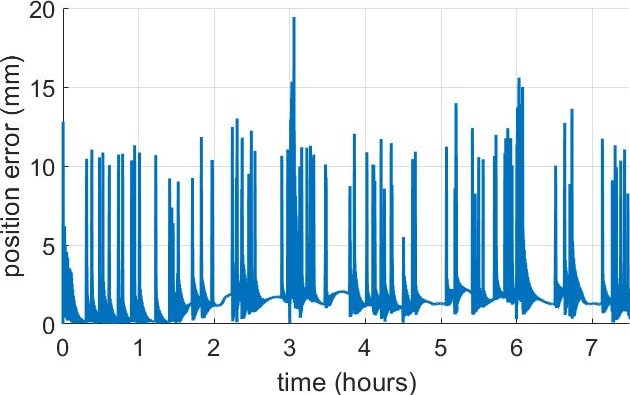}
        \caption{Position error between the reference and tracking trajectories.}
        \label{fig:pos comp w bang 8s}
    \end{subfigure}\hfill
    \begin{subfigure}[]{0.49\textwidth}
        \includegraphics[scale = 0.35]{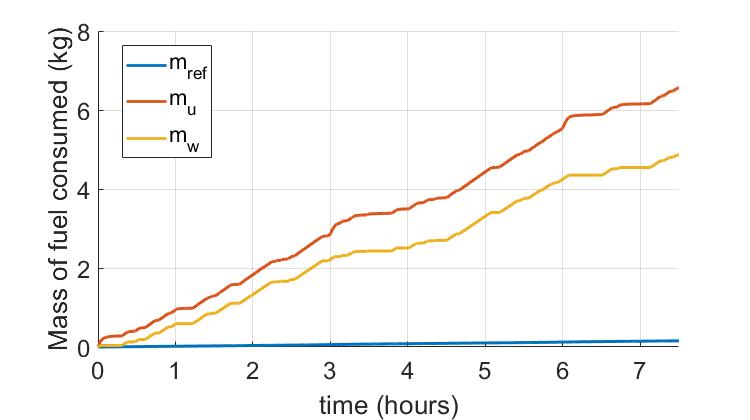}
        \caption{Comparisons of fuel used to complete the reference trajectory $m_\text{ref}$ (blue), the tracking trajectory $m_u$ (red), and the fuel used by malfunctioning thruster no. 4 $m_w$ (yellow).}
        \label{fig:fuel comp w bang 8s}
    \end{subfigure}
    \caption{Analysis of the trajectory tracking performance for a bang-bang undesirable thrust signal $w$ and actuation delay $\tau = 8\, s$.}
    \label{fig:w bang 8s}
\end{figure}

The average position error with respect to the reference trajectory is $1.76\, mm$ and the maximal error is $19.4\, mm$. These values represent approximately a fourfold increase compared to the scenario of Section~\ref{subsec: w bang 1s}. As in the Lipschitz cases where we also witnessed a fourfold increase between $\tau = 0.2\, s$ and $\tau = 8\, s$, the increased actuation delay has significant impact on the tracking accuracy.

The undesirable thrust still consumes $m_w = 4.86\, kg$ of fuel, but the controller now needs $m_u = 6.56\, kg$ according to Fig.~\ref{fig:w bang 8s}(\subref{fig:fuel comp w bang 8s}) instead of $5.33\, kg$ for $\tau = 1\, s$. This consumption increase is most likely caused by the large thrust spikes of Fig.~\ref{fig:w bang inputs}(\subref{fig:inputs w bang 8s}).
As in the Lipschitz case, the increased actuation delay does not have a significant impact on the fuel consumption. However, if we increase $\tau$ to $10\, s$, then the situation is similar as that of Section~\ref{subsec: w rand 10s} with a prohibitive increase in fuel consumption to keep the malfunctioning spacecraft close to the reference orbit.

In all scenarios tested so far, the undesirable thrust signal was saturated at $1\%$ of its capability to have $\|w\|$ of the same order of magnitude as $\|u_\text{ref}\|$ as depicted on Fig.~\ref{fig:inputs}(\subref{fig:total_inputs}). In the first scenario $\|u\|$ was also of the same order of magnitude. However, we see in this scenario that $u$ sometimes needs to significantly overshoot $w$. Therefore, we must also investigate the scenario where $w$ has access to its whole thrust capability, i.e., $w(t) \in [0, 1]$, to assess whether it can be counteracted by $u$ despite its saturation limit.

\subsection{Saturated Lipschitz undesirable thrust and actuation delay of 2 second}\label{subsec: rand w=1}

We will now investigate the case of a Lipschitz undesirable thrust input $w$ where $L = 0.1$, $\max\{ w(t),\, t \geq 0\} = 1$ and $\tau = 2\, s$. Since $w$ makes use of its full range of thrust actuation, the controlled thrusters might reach their own saturation limit. The simulation results show the undesirable thrust signal meeting both its saturation limits, $w(t) \in [0, 1]$ as seen on Fig.~\ref{fig:inputs thrusters w rand 1}(\subref{fig:inputs w rand 1}). The controlled thrusters however, do not reach their own saturation since the individual magnitude of each thruster never reaches $1$ on Fig.~\ref{fig:inputs thrusters w rand 1}(\subref{fig:thrusters w rand 1}), except at 1 hour 30 minutes. This saturation can also be seen on Fig.~\ref{fig:inputs thrusters w rand 1}(\subref{fig:inputs w rand 1}) where $\|u\| = \sqrt{2}$.

\begin{figure}[htbp!]
    \centering
    \begin{subfigure}[]{0.49\textwidth}
        \includegraphics[scale = 0.45]{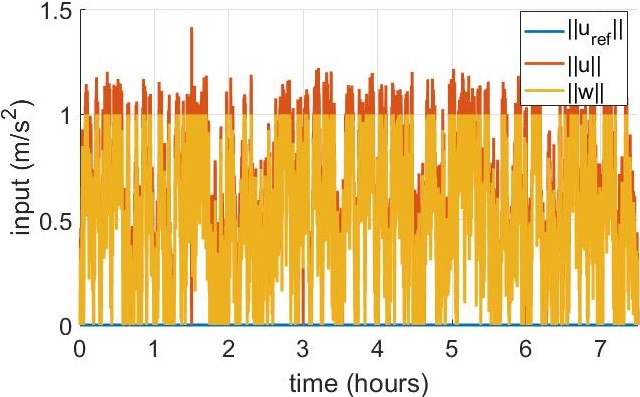}
        \caption{Comparison of input magnitude between the reference $\|u_\text{ref}\|$ (blue), tracking control $\|u\|$ (red) and the saturated undesirable input $\|w\|$ (yellow).}
        \label{fig:inputs w rand 1}
    \end{subfigure}\hfill
    \begin{subfigure}[]{0.49\textwidth}
        \includegraphics[scale = 0.35]{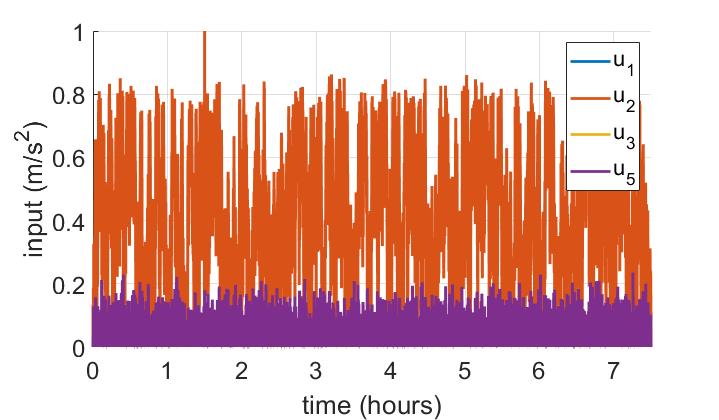}
        \caption{Thrust profiles for the four controlled thrusters of the chaser satellite on the tracking trajectory.}
        \label{fig:thrusters w rand 1}
    \end{subfigure}
    \caption{Analysis of the thrust inputs for a Lipschitz undesirable thrust signal $w$ and actuation delay $\tau = 2\, s$.}
    \label{fig:inputs thrusters w rand 1}
\end{figure}

Based on Fig.~\ref{fig:inputs thrusters w rand 1}(\subref{fig:thrusters w rand 1}), we can see that thruster no. 2 is producing the thrust necessary to counteract $w$. Thruster no. 1 is actually matching $u_2$, just as in Fig.~\ref{fig:inputs}(\subref{fig:thrusters}), except that we cannot see it on Fig.~\ref{fig:inputs thrusters w rand 1}(\subref{fig:thrusters w rand 1}) because $u_2$ covers $u_1$.
The average position error is contained to $48\, mm$, while the maximal position error is $0.29\, m$ as shown on Fig.~\ref{fig:pos velocity w rand 1}(\subref{fig:pos comp w rand 1}). Then, the tracking trajectory stays sufficiently close to the reference to not be distinguishable on a figure like Fig.~\ref{fig:traj and error}(\subref{fig:tracking comparison}). The tracking velocity presents large fluctuations above the reference velocity as shown on Fig.~\ref{fig:velocity w rand 1}(\subref{fig:velocity w rand 1}) while staying much closer than in the scenario of Section~\ref{subsec: w rand 10s} where $v$ was entirely above $v_\text{ref}$ as seen on Fig.~\ref{fig: w rand 10s}(\subref{fig:velocity w rand 10s}).

\begin{figure}[htbp!]
    \centering
    \begin{subfigure}[]{0.49\textwidth}
        \includegraphics[scale = 0.45]{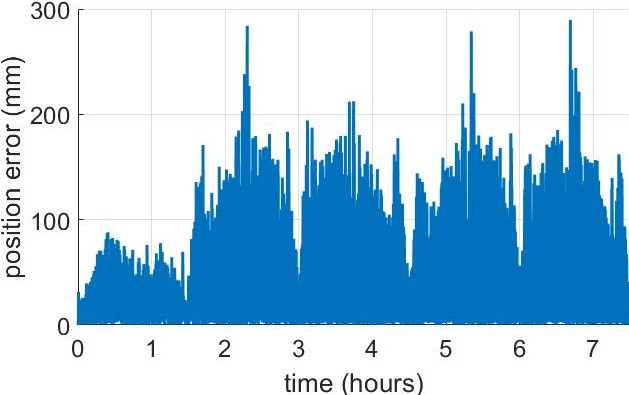}
        \caption{Position error between the reference and tracking trajectories.}
        \label{fig:pos comp w rand 1}
    \end{subfigure}\hfill
    \begin{subfigure}[]{0.49\textwidth}
        \includegraphics[scale = 0.45]{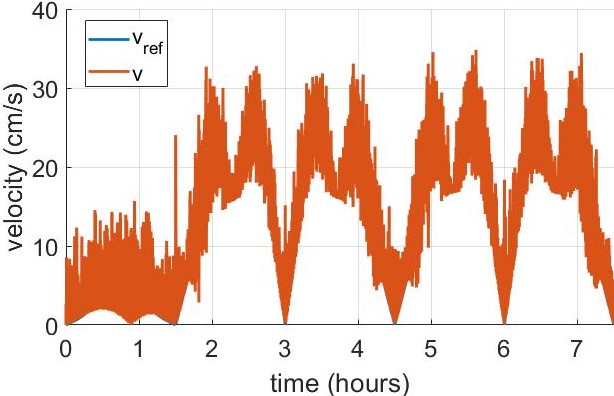}
        \caption{Velocity comparison for the reference trajectory $v_\text{ref}$ (blue) and the tracking trajectory $v$ (red).}
        \label{fig:velocity w rand 1}
    \end{subfigure}
    \caption{Analysis of the trajectory tracking performance for a Lipschitz undesirable thrust signal $w$ and actuation delay $\tau = 2\, s$.}
    \label{fig:pos velocity w rand 1}
\end{figure}

Because of the large thrusts employed in this scenario, the masses of fuel consumed have also significantly increased. The controlled thrusters would need $m_u = 360\, kg$ of fuel, while the malfunctioning thruster is guzzling $m_w = 342\, kg$ of fuel over the 7.5 hours of the mission. These masses are relatively close, within $5\%$ of each other, which tells us that the controller is not wasting too much extra fuel in overshoots, it uses only what is needed to counteract $w$.
However, recall that our spacecraft mass was set at $600\, kg$. Thus, if such a malfunction were to happen, the thrusters would run out of fuel before completing the mission. 
Nevertheless, while fuel is available, we now know that controller~\eqref{eq: feedback controller} can compensate time-varying undesirable thrust of maximal amplitude.

With the same Lipschitz undesirable thrust signal, but an actuation delay $\tau = 3\, s$ instead of $2\, s$, the trajectory quickly diverge from the reference. This was somewhat predictable from the saturation of $u_2$ in Fig.~\ref{fig:inputs thrusters w rand 1}(\subref{fig:thrusters w rand 1}).
Let us now study how controller~\eqref{eq: feedback controller} would fare against a bang-bang undesirable input of similar magnitude.

\subsection{Saturated bang-bang undesirable thrust and actuation delay of 1 second}\label{subsec: bang w=1}

In this scenario $w$ is bang-bang in $[0, 1]$ and the actuation delay is $\tau = 1\, s$. The simulation shows clearly the bang-bang behavior of the undesirable thrust signal on Fig.~\ref{fig:inputs thrusters w bang 1}(\subref{fig:inputs w bang 1}). The controlled thrusters are also reaching their own saturation limit of $1$ on Fig.~\ref{fig:inputs thrusters w bang 1}(\subref{fig:thrusters w bang 1}), except at 1 hour 30 minutes. This saturation can also be seen on Fig.~\ref{fig:inputs thrusters w rand 1}(\subref{fig:inputs w rand 1}) where $\|u\| = \sqrt{2}$.

\begin{figure}[htbp!]
    \centering
    \begin{subfigure}[]{0.49\textwidth}
        \includegraphics[scale = 0.45]{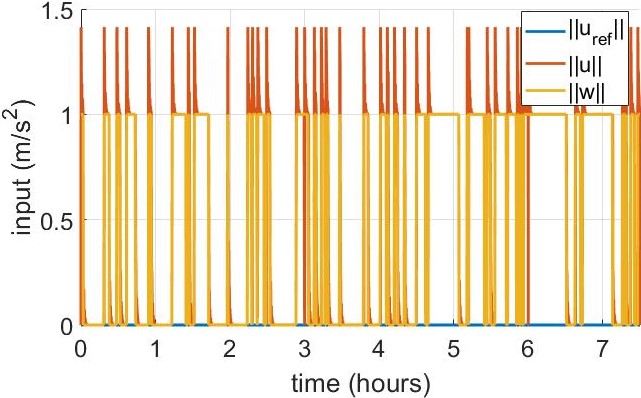}
        \caption{Comparison of input magnitude between the reference $\|u_\text{ref}\|$ (blue), tracking control $\|u\|$ (red) and the bang-bang undesirable input $\|w\|$ (yellow).}
        \label{fig:inputs w bang 1}
    \end{subfigure}\hfill
    \begin{subfigure}[]{0.49\textwidth}
        \includegraphics[scale = 0.45]{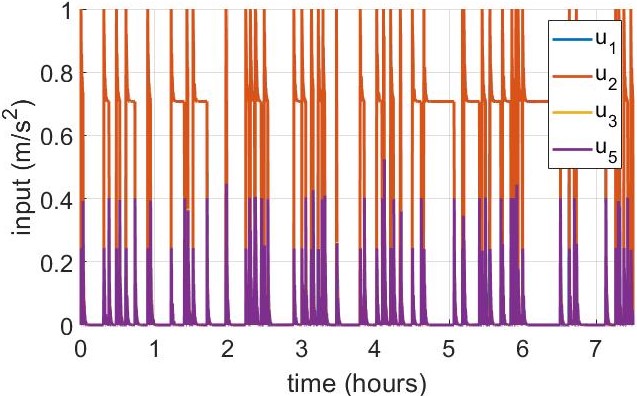}
        \caption{Thrust profiles for the four controlled thrusters of the chaser satellite on the tracking trajectory.}
        \label{fig:thrusters w bang 1}
    \end{subfigure}
    \caption{Analysis of the trajectory tracking performance for a bang-bang undesirable thrust signal $w$ and actuation delay $\tau = 1\, s$.}
    \label{fig:inputs thrusters w bang 1}
\end{figure}

As shown on Fig.~\ref{fig:pos velocity w bang 1}(\subref{fig:pos comp w bang 1}), the average position error is $17.1\, mm$ and the maximal position error is $0.5\, m$, so both trajectories are still indistinguishable on a figure like Fig.~\ref{fig:traj and error}(\subref{fig:tracking comparison}).
We note the presence of a velocity spike on Fig.~\ref{fig:pos velocity w bang 1}(\subref{fig:velocity w bang 1}) for each spike of $w$ on Fig.~\ref{fig:inputs thrusters w bang 1}(\subref{fig:inputs w bang 1}).

\begin{figure}[htbp!]
    \centering
    \begin{subfigure}[]{0.49\textwidth}
        \includegraphics[scale = 0.45]{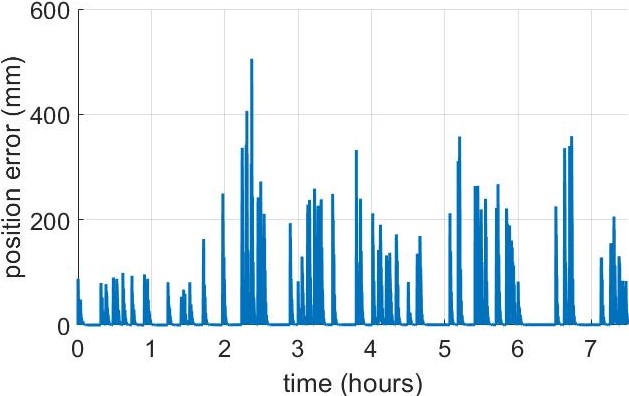}
        \caption{Position error between the reference and tracking trajectories.}
        \label{fig:pos comp w bang 1}
    \end{subfigure}\hfill
    \begin{subfigure}[]{0.49\textwidth}
        \includegraphics[scale = 0.45]{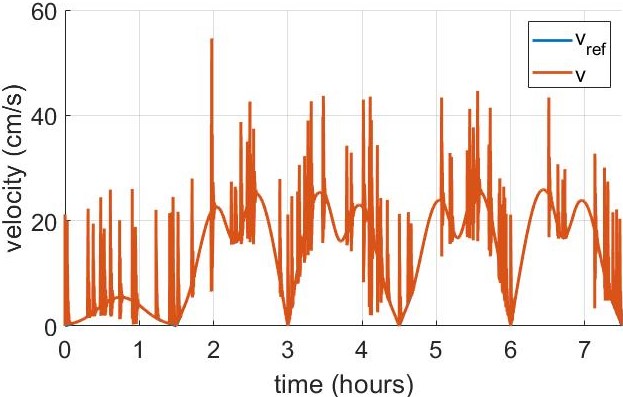}
        \caption{Velocity comparison for the reference trajectory $v_\text{ref}$ (blue) and the tracking trajectory $v$ (red).}
        \label{fig:velocity w bang 1}
    \end{subfigure}
    \caption{Analysis of the trajectory tracking performance for a bang-bang undesirable thrust signal $w$ and actuation delay $\tau = 1\, s$.}
    \label{fig:pos velocity w bang 1}
\end{figure}

As in the previous scenario, the fuel consumption is too large for the mission to be completed with such a malfunctioning thruster, but while it is active it can be actively counteracted by $u_1$ and $u_2$ as shown on Fig.~\ref{fig:inputs thrusters w bang 1}(\subref{fig:thrusters w bang 1}). The masses of fuel consumed by $u$ and $w$ are also relatively close, within $3\%$ of each other, with $m_u = 345\, kg$ and $m_w = 336\, kg$, which relates to the efficiency of the controller.

If we further increase the actuation delay to $\tau = 2\, s$ for the same undesirable thrust $w$, the trajectory quickly diverge from the reference. Since the controlled inputs were already saturated for $\tau = 1\, s$ as seen on Fig.~\ref{fig:inputs thrusters w bang 1}(\subref{fig:thrusters w bang 1}), the controller was not able to overcome a more unpredictable $w$ and this divergence is not surprising.

\subsection{Summary of the simulation scenarios}

We summarize the scenarios studied above in Table~\ref{tab: summary}.
For each scenario, we compute the average and maximal position errors, the mass of fuel used by the controlled thrusters $m_u$ and by malfunctioning thruster no. 4 $m_w$. We also calculate the relative difference of fuel used $\frac{m_u - m_w - m_\text{ref}}{m_w + m_\text{ref}}$, which is a good metric for the efficiency of controller~\eqref{eq: feedback controller} in overcoming $w$ without excessive thrust. 
\begin{table}[htbp!]
    \centering
    \def\arraystretch{1.2}
    \begin{tabular}{m{19mm}m{15mm}m{17mm}|m{15mm}m{17mm}m{23mm}m{20mm}m{21mm}}
        Regularity of $w$ & Actuation delay $\tau$ & Saturation of $w$ & Average position error & Maximal position error & Fuel used by controlled thrusters $m_u$ & Fuel used by thruster no.4 $m_w$ & Relative difference of fuel used\\ \hline
        Lipschitz & $0.2\, s$ & $0.01$ & $0.36\, mm$ & $1.05\, mm$ & $1.31\, kg$ & $1.06\, kg$ & $7.4\%$ \\
        Lipschitz & $8\, s$ & $0.01$ & $1.2\, mm$ & $4.1\, mm$ & $1.38\, kg$ & $1.06\, kg$ & $13.1\%$ \\
        Lipschitz & $10\, s$ & $0.01$ & $484\, mm$ & $3103\, mm$ & $503\, kg$ & $1.06\, kg$ & $41130\%$\\
        bang-bang & $1\, s$ & $0.01$ & $0.54\, mm$ & $5.6\, mm$ & $5.33\, kg$ & $4.86\, kg$ & $6.2\%$ \\
        bang-bang & $8\, s$ & $0.01$ & $1.76\, mm$ & $19.4\, mm$ & $6.56\, kg$ & $4.86\, kg$ & $31\%$ \\
        Lipschitz & $2\, s$ & $1$ & $48\, mm$ & $292\, mm$ & $360\, kg$ & $342\, kg$ & $5.2\%$ \\
        bang-bang & $1\, s$ & $1$ & $17.1\, mm$ & $509\, mm$ & $345\, kg$ & $336\, kg$ & $2.6\%$
    \end{tabular}
    \caption{Summary of the simulation scenarios.}
    \label{tab: summary}
\end{table}

Let us now summarize the findings of the various scenarios studied.
Despite the narrow range of application of Theorem~\ref{thm: closed loop res track}, controller~\eqref{eq: feedback controller} provides tracking accuracy to the millimeter scale on a wide range of scenarios with fast-varying undesirable inputs and long actuation delays.
A potential issue limiting the application of controller~\eqref{eq: feedback controller} was the expected saturation of $u$ when $w$ reaches its maximal amplitude. However, scenarios of Section~\ref{subsec: rand w=1} and \ref{subsec: bang w=1} showed that for small actuation delays a maximal $w$ can still be efficiently counteracted. 
We can visually summarize the performance of controller~\eqref{eq: feedback controller} with the Pareto front of Fig.~\ref{fig:pareto} describing the maximal value of $w$ allowing successful tracking at a given actuation delay $\tau$. Based on the scenarios investigated in this section, we decided to consider the tracking to be successful when the position error is always smaller than $0.8\, m$, which is $1\%$ of the minimal target distance on the reference trajectory.

\begin{figure}[htbp!]
    \centering
    \includegraphics[scale = 0.55]{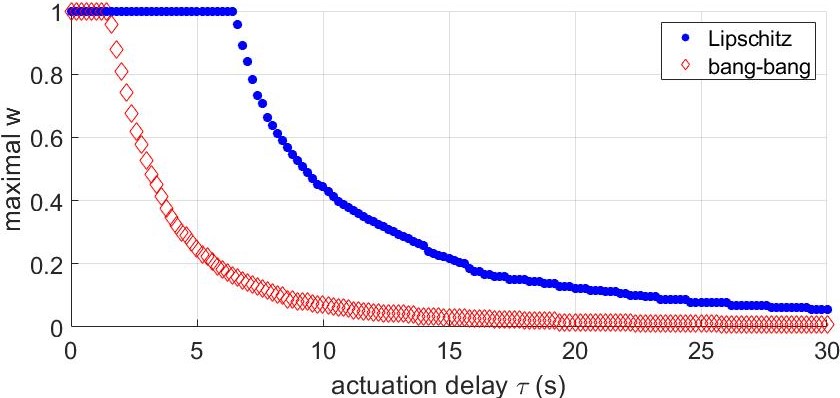}
    \caption{Pareto front of the maximal saturation limit of $w$ for which controller~\eqref{eq: feedback controller} maintains a position error under $0.8\, m$ despite actuation delay $\tau$.}
    \label{fig:pareto}
\end{figure}

To conclude, after the loss of control authority over thruster no. 4, controller~\eqref{eq: feedback controller} ensures completion of the inspection mission for any undesirable thrust signal $w$ as long as the actuation delay is inferior to $1\, s$.

\section{Conclusion and future work}

We presented a new methodology to safely perform a satellite inspection mission despite the loss of control authority over a thruster of the inspecting satellite. The controller of this malfunctioning spacecraft is further plagued by a constant actuation delay. To mitigate partial loss of control authority and actuation delays, we developed resilience theory for linear systems with actuation delay before extending these results to the nonlinear dynamics of the spacecraft under study. We established analytical trajectory tracking guarantees on a resilient feedback controller embedded with a state predictor to compensate for the actuation delay. This controller enables a resilient tracking of the reference trajectory and a safe completion of the inspection mission even when the uncontrolled thruster produces maximal bang-bang inputs.

There are several promising avenues of future work. First, we want to implement our controller on a high-fidelity simulator to study its robustness to unmodeled dynamics such as the nonlinear terms neglected in the Clohessy-Wiltshire framework \citep{Ortolano}. 
We also have the objective of extending resilience theory to nonlinear dynamics by deriving a new proof for H\'ajek's duality theorem \citep{Hajek}. This extension would enable the study of spacecraft models combining position and attitude for a more realistic treatment.

\section*{Acknowledgment}

This work was supported by an Early Stage Innovations grant from NASA’s Space Technology Research Grants Program, grant no. 80NSSC19K0209, and by NASA grant no. 80NSSC21K1030.

\bibliographystyle{IEEEtran}
\bibliography{ref}

\end{document}